\def\year{2022}\relax
\documentclass[letterpaper]{article} 
\usepackage{aaai22}  
\usepackage{times}  
\usepackage{helvet}  
\usepackage{courier}  
\usepackage[hyphens]{url}  
\usepackage{graphicx} 
\usepackage{natbib}  
\usepackage{caption} 
\DeclareCaptionStyle{ruled}{labelfont=normalfont,labelsep=colon,strut=off} 
\usepackage{amsfonts}
\usepackage{amssymb,amsthm,amsmath}
\usepackage{esint}
\usepackage{isomath}
\usepackage{latexsym}
\usepackage{epic}
\usepackage{epsfig}
\usepackage{multirow}
\usepackage{hhline}
\usepackage{wrapfig}
\usepackage{verbatim}
\usepackage{enumitem}
\usepackage{color}
\usepackage{thmtools}
\usepackage{thm-restate}
\usepackage[ruled,vlined, linesnumbered, boxed]{algorithm2e}
\usepackage[labelformat=simple]{subcaption}

\usepackage{cleveref}
\usepackage{footmisc}
\usepackage{mathtools}
\usepackage{graphicx}

\newcommand{\R}{\mathbb{R}}

\newtheorem{theorem}{Theorem}
\newtheorem{definition}{Definition}

\newtheorem{proposition}{Proposition}

\newcommand{\decision}{x}
\newcommand{\decisionset}{\mathcal{X}}

\newcommand{\environment}{\pi}
\newcommand{\environmentset}{\Pi}

\newcommand{\norm}[1]{\left\lVert#1\right\rVert}

\title{Coordinating Followers to Reach Better Equilibria: End-to-End Gradient Descent for Stackelberg Games}
\author {
    Kai Wang,\textsuperscript{\rm 1}
    Lily Xu,\textsuperscript{\rm 1}
    Andrew Perrault,\textsuperscript{\rm 2}
    Michael K. Reiter,\textsuperscript{\rm 3}
    Milind Tambe\textsuperscript{\rm 1}
}
\affiliations {
    \textsuperscript{\rm 1}Harvard University, 
    \textsuperscript{\rm 2}The Ohio State University, 
    \textsuperscript{\rm 3}Duke University\\
    \{kaiwang,lily\_xu\}@g.harvard.edu, perrault.17@osu.edu,
    michael.reiter@duke.edu, milind\_tambe@harvard.edu
}

\usepackage{bibentry}

\begin{document}

\maketitle

\begin{abstract}
A growing body of work in game theory extends the traditional Stackelberg game to settings with one leader and multiple followers who play a Nash equilibrium. Standard approaches for computing equilibria in these games reformulate the followers' best response as constraints in the leader's optimization problem. These reformulation approaches can sometimes be effective, but make limiting assumptions on the followers' objectives and the equilibrium reached by followers, e.g., uniqueness, optimism, or pessimism. To overcome these limitations, we run gradient descent to update the leader's strategy by differentiating through the equilibrium reached by followers. Our approach generalizes to any stochastic equilibrium selection procedure that chooses from multiple equilibria, where we compute the stochastic gradient by back-propagating through a sampled Nash equilibrium using the solution to a partial differential equation to establish the unbiasedness of the stochastic gradient. Using the unbiased gradient estimate, we implement the gradient-based approach to solve three Stackelberg problems with multiple followers. Our approach consistently outperforms existing baselines to achieve higher utility for the leader.
\end{abstract}
\section{Introduction}
Stackelberg games are commonly adopted in many real-world applications, including security~\cite{jiang2013defender,gan2020mechanism}, wildlife conservation~\cite{fang2016deploying}, and commercial decisions made by firms~\cite{naghizadeh2014voluntary,aussel2020trilevel,zhang2016multi}.
Moreover, many realistic settings involve a single leader with multiple self-interested followers such as wildlife conservation efforts with a central coordinator and a team of defenders~\cite{gan2018stackelberg,gan2020mechanism}; resource management in energy~\cite{aussel2020trilevel} with suppliers, aggregators, and end users; or security problems with a central insurer and a set of vulnerable agents~\cite{naghizadeh2014voluntary,johnson2011security}. Solving Stackelberg games with multiple followers is challenging in general~\cite{basilico2017methods,coniglio2020computing}. Previous work often reformulates the followers' best response as stationary and complementarity constraints in the leader's optimization~\cite{shi2005kth,basilico2020bilevel,basilico2017methods,coniglio2020computing,calvete2007linear}, casting the entire Stackelberg problem as a single optimization problem.
This reformulation approach has achieved significant success in problems with linear or quadratic objectives, assuming a unique equilibrium or a specific equilibrium concept, e.g., followers' optimistic or pessimistic choice of equilibrium ~\cite{hu2011variational,basilico2020bilevel,basilico2017methods}.
The reformulation approach thoroughly exploits the structure of objectives and equilibrium to conquer the computation challenge.
However, when these conditions are not met, reformulation approach may get trapped in low-quality solutions.

In this paper, we propose an end-to-end gradient descent approach to solve multi-follower Stackelberg games.
Specifically, we run gradient descent by back-propagating through a sampled Nash equilibrium reached by followers to update the leader's strategy.
Our approach overcomes weaknesses of reformulation approaches as (i)~we decouple the leader's optimization problem from the followers', casting it as a learning problem to be solved by end-to-end gradient descent through the followers' equilibrium; and (ii)~back-propagating through a sampled Nash equilibrium enables us to work with arbitrary equilibrium selection procedures and multiple equilibria.

In short, we make several contributions.
First, we provide a procedure for differentiating through a Nash equilibrium assuming uniqueness (later we relax the assumption). Because each follower must simultaneously best respond to every other follower, the Karush--Kuhn--Tucker (KKT) conditions~\cite{kuhn2014nonlinear} for each follower must be simultaneously satisfied. We can thus differentiate through the system of KKT conditions and apply the implicit function theorem to obtain the gradient. Second, we relax the uniqueness assumption and extend our approach to an arbitrary, potentially stochastic, equilibrium selection oracle. We first show that given a stochastic equilibrium selection procedure, using optimistic or pessimistic assumptions to solve Stackelberg games with stochastic equilibria can yield payoff to the leader that is arbitrarily worse than optimal. 
To address the issue of multiple equilibria and stochastic equilibria, we formally characterize stochastic equilibria with a concept we call \emph{equilibrium flow}, defined by a partial differential equation. 
The equilibrium flow ensures the stochastic gradient computed from the sampled Nash equilibrium is unbiased, allowing us to run stochastic gradient descent to differentiate through the stochastic equilibrium.
We also discuss how to compute the equilibrium flow either from KKT conditions under certain sufficient conditions or by solving the partial differential equation.
This paper is the first to guarantee that the gradient computed from an arbitrary stochastic equilibrium sampled from multiple equilibria is a differentiable, unbiased sample.
Third, to address the challenge that the feasibility of the leader's strategy may depend on the equilibrium reached by the followers (e.g., when a subsidy paid to the followers is conditional on their actions as in~\cite{rotemberg2019equilibrium,mortensen2001taxes}), we use an augmented Lagrangian method to convert the constrained optimization problem into an unconstrained one.
The Lagrangian method combined with our unbiased Nash equilibrium gradient estimate enables us to run stochastic gradient descent to optimize the leader's payoff while also satisfying the equilibrium-dependent constraints.

We conduct experiments to evaluate our approach in three different multi-follower Stackelberg games: normal-form games with a leader offering subsidies to followers, Stackelberg security games with a planner coordinating multiple defenders, and cyber insurance games with an insurer and multiple customers.
Across all three examples, the leader's strategy space is constrained by a budget constraint that depends on the equilibrium reached by the followers.
Our gradient-based method provides a significantly higher payoff to the leader evaluated at equilibrium, compared to existing approaches which fail to optimize the leader's utility and often produce large constraint violations. These results, combined with our theoretical contributions, demonstrate the strength of our end-to-end gradient descent algorithm in solving Stackelberg games with multiple followers.

\section*{Related Work}


\paragraph{Stackelberg models with multiple followers}
Multi-follower Stackelberg problems have received a lot of attention in domains with a hierarchical leader-follower structure~\cite{nakamura2015one,zhang2016multi,liu1998stackelberg,solis2016modeling,sinha2014finding}.
Although single-follower normal-form Stackelberg games can be solved in polynomial time~\cite{korzhyk2010complexity,blum2019computing}, the problem becomes NP-hard when multiple followers are present, even when the equilibrium is assumed to be either optimistic or pessimistic~\cite{basilico2020bilevel,coniglio2020computing}.
Existing approaches~\cite{basilico2020bilevel,aussel2020trilevel} primarily leverage the leader-follower structure in a bilevel optimization formulation~\cite{colson2007overview}, which can be solved by reformulating the followers' best response into non-convex stationary and complementarity constraints in the leader's optimization problem~\cite{sinha2019using}.
Various optimization techniques, including branch-and-bound~\cite{coniglio2020computing} and mixed-integer programs~\cite{basilico2020bilevel}, are adopted to solve the reformulated problems. 
However, these reformulation approaches highly rely on well-behaved problem structure, which may encounter large mixed integer non-linear programs when the followers have non-quadratic objectives.
Additionally, these approaches mostly assume uniqueness of equilibrium or a specific equilibrium concept, e.g., optimistic or pessimistic, which may not be feasible~\cite{gan2018stackelberg}.
Previous work on the stochastic equilibrium drawn from multiple equilibria in Stackelberg problems~\cite{lina1996hierarchical} mainly focuses on the existence of an optimal solution, while our work focuses on actually solving the Stackelberg problems to identify the best action for the leader.

In contrast, our approach solves the Stackelberg problem by differentiating through the equilibrium reached by followers to run gradient descent in the leader's problem. Our approach also applies to any stochastic equilibrium drawn from multiple equilibria by establishing the unbiasedness of the gradient computed from a sampled equilibrium using a partial differential equation.

\paragraph{Differentiable optimization}
When there is only a single follower optimizing his utility function, differentiating through a Nash equilibrium reduces to the framework of differentiable optimization~\cite{pirnay2012optimal,amos2017optnet,agrawal2019differentiable,bai2019deep}.
When there are two followers with conflicting objectives (zero-sum), differentiating through a Nash equilibrium reduces to a differentiable minimax formulation~\cite{ling2018game,ling2019large}.
Lastly, when there are multiple followers,~\citeauthor{li2020end}~\shortcite{li2020end} follow the sensitivity analysis and variational inequalities (VIs) literature~\cite{mertikopoulos2019learning,ui2016bayesian,dafermos1988sensitivity,parise2019variational} to express a unique Nash equilibrium as a fixed-point to the projection operator in VIs to differentiate through the equilibrium. ~\citeauthor{li2021solving}~\shortcite{li2021solving} further extend the same approach to structured hierarchical games.
Nonetheless, these approaches rely on the uniqueness of Nash equilibrium. In contrast, our approach generalizes to multiple equilibria.

\section{Stackelberg Games With a Single Leader and Multiple Followers}\label{sec:model}
In this paper, we consider a Stackelberg game composed of one leader and $n$ followers.
The leader first chooses a strategy $\environment \in \environmentset$ that she announces, then the followers observe the leader's strategy and respond accordingly.
When the leader's strategy $\environment$ is determined, the followers form an $n$-player simultaneous game with $n$~followers, where the $i${\nobreakdash-}th follower minimizes his own objective function $f_i(\decision_i, \decision_{-i}, \environment)$, which depends on his own action $\decision_i \in \decisionset_i$, other followers' actions $\decision_{-i} \in \decisionset_{-i}$, and the leader's strategy $\environment \in \environmentset$.
We assume that each strategy space is characterized by linear constraints: $\decisionset_i = \{ \decision_i \mid A_i \decision_i = b_i, G_i \decision_i \leq h_i \}$.
We also assume perfect information---all the followers know other followers' utility functions and strategy spaces.

\subsection{Nash Equilibria}
We call $\boldsymbol\decision^* = \{ \decision_1^*, \decision_2^*, \ldots, \decision_n^* \}$ a Nash equilibrium if no follower has an incentive to deviate from their current strategy, where we assume each follower \textit{minimizes}\footnote{We use minimization formulation to align with the convention in convex optimization. In our experiments, examples of maximization problems are used, but the same approach applies.} his objective:
\begin{align}\label{eqn:nash}
    & \forall i: f_i(\decision^*_i, \decision_{-i}^*, \environment) \leq f_i(\decision_i, \decision_{-i}^*, \environment) & \forall \decision_i \in \decisionset_i.
\end{align}
As shown in Figure~\ref{fig:framework}, when the leader's strategy $\environment$ is chosen and passed to an $n$-player game composed of all followers, we assume the followers converge to a Nash equilibrium $\boldsymbol\decision^*$. 

In the first section, we assume there is a unique Nash equilibrium returned by an oracle $\boldsymbol\decision^* = \mathcal{O}(\environment)$.
We later generalize to the case where there are multiple equilibria with a stochastic equilibrium selection oracle which randomly outputs an equilibrium $\boldsymbol\decision \sim \mathcal{O}(\environment)$ drawn from a distribution with probability density function $p(\cdot, \environment): \decisionset \rightarrow \R_{\geq 0}$. 

\subsection{Leader's Optimization Problem}
When the leader chooses a strategy $\environment$ and all the followers reach an equilibrium $\boldsymbol\decision^*$, the leader receives a payoff $f(\boldsymbol\decision^*, \environment)$ and a constraint value $g(\boldsymbol\decision^*, \environment)$.
The goal of the Stackelberg leader is to choose an optimal $\environment$ to maximize her utility while satisfying the constraint.

\begin{definition}[Stackelberg problems with multiple followers and unique Nash equilibrium]\label{def:optimization-problem}
The leader chooses a strategy $\environment$ to maximize her utility function~$f$ subject to constraints~$g$ evaluated at the unique equilibrium $\boldsymbol\decision^*$ induced by an equilibrium oracle $\mathcal{O}$, i.e.,:
\begin{align}\label{eqn:environment-optimization}
    & \max\nolimits_{\environment} f(\boldsymbol\decision^*, \environment) & \text{s.t.} \quad \boldsymbol\decision^* = \mathcal{O}(\environment), \quad g(\boldsymbol\decision^*, \environment) \leq 0.
\end{align}
\end{definition}
This problem is hard because the objective $f(\boldsymbol\decision^*, \environment)$ depends on the Nash equilibrium $\boldsymbol\decision^*$ reached by the followers.
Moreover, notice that the feasibility constraint $g(\boldsymbol\decision^*, \environment)$ also depends on the equilibrium, which creates a complicated feasible region for the leader's strategy $\environment$. 



\begin{figure}[t]
    \centering
    \includegraphics[width=\linewidth]{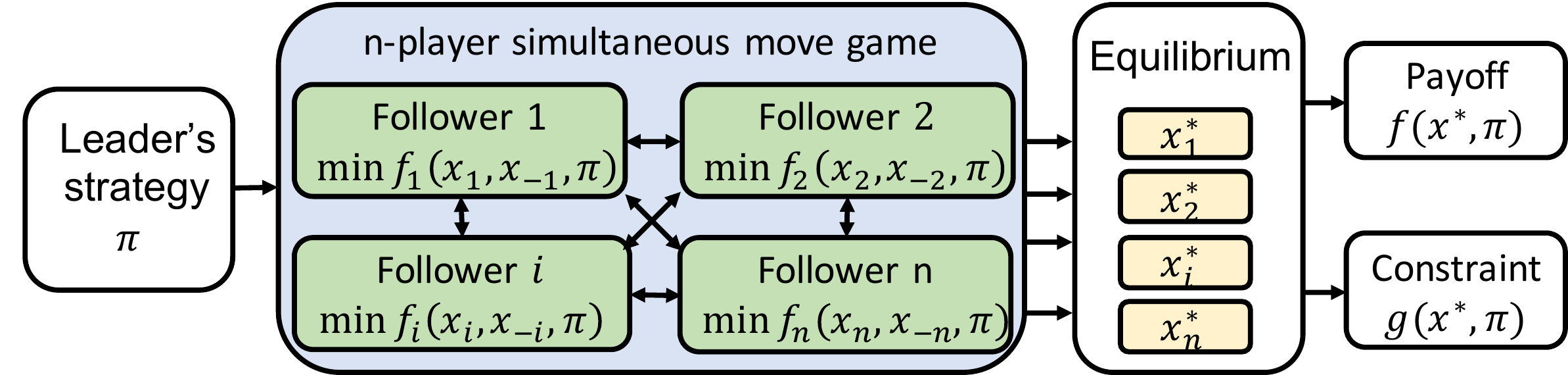}
    \caption{Given leader's strategy $\environment$, followers respond to the leader's strategy and reach a Nash equilibrium $\boldsymbol\decision^*$. The leader's payoff and the constraint depend on both the leader's strategy $\environment$ and the equilibrium $\boldsymbol\decision^*$.}
    \label{fig:framework}
\end{figure}

\subsection{Gradient Descent Approach}
To solve the leader's optimization problem, we propose to run gradient descent to optimize the leader's objective.
This requires us to compute the following gradient:
\begin{align}~\label{eqn:chain-rule}
\frac{d f(\boldsymbol\decision^*, \environment)}{d \environment} = f_\environment(\boldsymbol\decision^*, \environment) + f_{\boldsymbol\decision}(\boldsymbol\decision^*, \environment
) \cdot \frac{d \boldsymbol\decision^*}{d \environment}.
\end{align}
The terms $f_\environment, f_{\boldsymbol\decision
}$ above are easy to compute since the payoff function~ $f$ is explicitly given.
The main challenge is to compute $\frac{d \boldsymbol\decision^*}{d \environment}$ because it requires estimating how the Nash equilibrium $\boldsymbol\decision^*$ reached by followers responds to any change in the leader's strategy $\environment$.

\section{Gradient of Unique Nash Equilibrium}\label{sec:method}
In this section, we assume a unique Nash equilibrium reached by followers. Motivated by the technique proposed by~\citeauthor{amos2017optnet}~\shortcite{amos2017optnet}, we show how to differentiate through multiple KKT conditions to derive the derivative of a Nash equilibrium.


\subsection{Differentiating Through KKT Conditions}
Given the leader's strategy $\environment$, we express the KKT conditions of follower~$i$ with dual variables $\lambda^*_i$ and $\nu^*_i$ by:
\begin{align}\label{eqn:kkts}
    \left\{ \begin{array}{ll}
    \nabla_{\decision_i} f_i(\decision^*_i, \decision^*_{-i}, \environment) + G_i^\top \lambda_i^* + A_i^\top \nu_i^* = 0 \\
    \text{Diag}(\lambda^*_i) (G_i \decision^*_i - h_i) = 0 \\
    A_i \decision^*_i = b_i.
    \end{array} \right.
\end{align}

We want to estimate the impact of $\environment$ on the resulting Nash equilibrium $\boldsymbol\decision^*$. Supposing the objective functions $f_i \in C^2$ are twice-differentiable, we can compute the total derivative of the the KKT system in Equation~\ref{eqn:kkts} written in matrix form:
\begin{align}
    & \begin{bmatrix}
        \nabla^2_{\decision_i \decision_i} f_i & \nabla^2_{\decision_{-i} \decision_{i}} f_i & G^\top_i & A^\top_i \\
        \text{Diag}(\lambda^*_i) G_i & 0 & \text{Diag}(G_i \decision^*_i \! - \! h_i) & 0 \\
        A_i & 0 & 0 & 0
    \end{bmatrix} \!
    \begin{bmatrix}
    d \decision^*_i \\
    d \decision^*_{-i} \\
    d \lambda^*_i \\
    d \nu^*_i
    \end{bmatrix} \nonumber \\
    & =
    \begin{bmatrix}
    - \nabla^2_{\environment \decision_i} f_i d \environment - d G_i^\top \lambda^*_i - d A_i^\top \nu^*_i \\
    - \text{Diag}(\lambda^*_i) (d G_i \decision_i^* - d h_i) \\
    d b_i - d A_i \decision_i^*
    \end{bmatrix} \nonumber .
\end{align}

Since we assume the constraint matrices are constant, $d G_i, d h_i, d A_i, d b_i$ can be ignored. We concatenate the linear system for every follower $i$ and move $d \environment$ to the denominator:
\begin{align}
    \begin{bmatrix}
        \! \nabla_{\boldsymbol\decision} F \!\! & \!\! G^\top \!\! & \!\! A^\top \\
        \! \text{Diag}({\boldsymbol\lambda^*}) G \! & \! \text{Diag}(G \boldsymbol\decision^* \!\! - \!\! h) \!\! & \!\! 0 \!\! \\
        \! A \!\! & \!\! 0 \!\! & \!\! 0 \!\!
    \end{bmatrix} \!
    \begin{bmatrix} \!
    \frac{d \boldsymbol\decision^*}{d \environment} \\
    \frac{d \boldsymbol\lambda^*}{d \environment} \\
    \frac{d \boldsymbol\nu^*}{d \environment} \!
    \end{bmatrix}
    \!\! = \!\!
    \begin{bmatrix}
    - \nabla_{\environment} F \\
    0 \\
    0
    \end{bmatrix}\label{eqn:kkt_derivative_closed_form}
\end{align}
where $F = [(\nabla_{\decision_1} f_1)^\top, \dots, (\nabla_{\decision_n} f_n)^\top]^\top$ is a column vector, and $G = \text{Diag}(G_1, G_2, \dots, G_n), A = \text{Diag}(A_1, A_2, \dots, A_n)$ are the diagonalized placement of a list of matrices. 
In particular, the KKT matrix on the left-hand side of Equation~\ref{eqn:kkt_derivative_closed_form} matches the sensitivity analysis of Nash equilibria using variational inequalities~\cite{facchinei2014solving,dafermos1988sensitivity}.

\begin{proposition}\label{prop:implicit-function-theorem}
When the Nash equilibrium is unique and the KKT matrix in Equation~\ref{eqn:kkt_derivative_closed_form} is invertible, the implicit function theorem holds and $\frac{d \boldsymbol\decision^*}{d \environment}$ can be uniquely determined by Equation~\ref{eqn:kkt_derivative_closed_form}.
\end{proposition}

Proposition~\ref{prop:implicit-function-theorem} ensures the sufficient conditions for applying Equation~\ref{eqn:kkt_derivative_closed_form} to compute $\frac{d \boldsymbol\decision^*}{d \environment}$. Under these sufficient conditions, we can compute Equation~\ref{eqn:chain-rule} using Equation~\ref{eqn:kkt_derivative_closed_form}.




\begin{figure*}
\begin{subfigure}{.2\textwidth}
  \centering
    \begin{tabular}{c|c|c|c|}
    \multicolumn{1}{c}{} & \multicolumn{3}{c}{\scriptsize{Follower 1}} \\
    \cline{2-4}
    \multirow{3}{*}{\rotatebox[origin=c]{90}{\scriptsize{Follower 2}}} 
    & $1$ & $0$ & $0$ \\ \cline{2-4}
    & $0$ & $1$ & $0$ \\ \cline{2-4}
    & $0$ & $0$ & $1$ \\ \cline{2-4}
    \end{tabular}
  \caption{Strategy 1 payoffs}
  \label{fig:strat1}
\end{subfigure}%
\hspace*{0.75em}
\begin{subfigure}{.2\textwidth}
  \centering
  \begin{tabular}{c|c|c|c|}
    \multicolumn{1}{c}{} & \multicolumn{3}{c}{\scriptsize{Follower 1}} \\
    \cline{2-4}
    \multirow{3}{*}{\rotatebox[origin=c]{90}{\scriptsize{Follower 2}}} 
    & $0$ & $1$ & $0$ \\ \cline{2-4}
    & $0$ & $0$ & $1$ \\ \cline{2-4}
    & $1$ & $0$ & $0$ \\ \cline{2-4}
  \end{tabular}
  \caption{Strategy 2 payoffs}
  \label{fig:strat2}
\end{subfigure}%
\hspace*{0.75em}
\begin{subfigure}{.22\textwidth}
  \centering
    \begin{tabular}{c|c|c|c|}
    \multicolumn{1}{c}{} & \multicolumn{3}{c}{\scriptsize{Follower 1}} \\
    \cline{2-4}
    \multirow{3}{*}{\rotatebox[origin=c]{90}{\scriptsize{Follower 2}}} 
    & $0$ & $0$ & $1$ \\ \cline{2-4}
    & $1$ & $0$ & $0$ \\ \cline{2-4}
    & $0$ & $1$ & $0$ \\ \cline{2-4}
    \end{tabular}
  \caption{Strategy 3 payoffs}
  \label{fig:strat3}
\end{subfigure}%
\hspace*{0.75em}
\begin{subfigure}{.3\textwidth}
  \centering
    \begin{tabular}{c|c|c|c|}
    \multicolumn{1}{c}{} & \multicolumn{3}{c}{\scriptsize{Follower 1}} \\
    \cline{2-4}
    \multirow{3}{*}{\rotatebox[origin=c]{90}{\scriptsize{Follower 2}}} 
    & $C$ & $0$ & $-\epsilon$ \\ \cline{2-4}
    & $C-\epsilon$ & $0$ & $0$ \\ \cline{2-4}
    & $0$ & $C-\epsilon$ & $-C$ \\ \cline{2-4}
    \end{tabular}
  \caption{Leader payoffs}
  \label{fig:strat-leader}
\end{subfigure}%
    \caption{Payoff matrices from Theorem~\ref{thm:stochastic-eq} where the leader has 3 strategies. Follower payoffs for each strategy in~(a)--(c) where both followers receive the same payoff; leader payoffs in~(d).
    }
    \label{fig:stochastic-eq-example}
\end{figure*}

\section{Gradient of Stochastic Equilibrium}\label{sec:multiple-equilibria}
In the previous section, we showed how to compute the gradient of a Nash equilibrium when the equilibrium is unique. However, this can be restrictive because Stackelberg games with multiple followers often have multiple equilibria that the followers can stochastically reach one. For example, both selfish routing games in the traffic setting~\cite{roughgarden2004stackelberg} and security games with multiple defenders~\cite{gan2018stackelberg} can have multiple equilibria that are reached in multiple independent runs.

In this section, we first demonstrate the importance of stochastic equilibrium by showing that optimizing over optimistic or pessimistic equilibrium could lead to arbitrarily bad leader's payoff under the stochastic setting. Second, we generalize our gradient computation to the case with multiple equilibria, allowing the equilibrium oracle $\mathcal{O}$ to stochastically return a sample equilibrium from a distribution of multiple equilibria. Lastly, we discuss how to compute the gradient of different types of equilibria and its limitation.

\subsection{Importance of Stochastic Equilibrium}
When the equilibrium oracle is stochastic, our Stackelberg problem becomes a stochastic optimization problem:
\begin{definition}[Stackelberg problems with multiple followers and stochastic Nash equilibria]\label{def:stochastic-problem}
The leader chooses a strategy $\environment$ to optimize her expected utility and satisfy the constraints in expectation under a given stochastic equilibrium oracle $\mathcal{O}$:
\begin{align}\label{eqn:expected-utility}
    \max\limits_{\environment} \mathop{\mathbb{E}}\limits_{\boldsymbol\decision^* \sim \mathcal{O}(\environment)} ~f(\boldsymbol\decision^*, \environment) \quad \text{s.t.} ~ \mathop{\mathbb{E}}\limits_{\boldsymbol\decision^* \sim \mathcal{O}(\environment)} ~g(\boldsymbol\decision^*, \environment) \leq 0.
\end{align}
\end{definition}

In particular, we show that if we ignore the stochasticity of equilibria by simply assuming optimistic or pessimistic equilibria, the leader's expected payoff can be arbitrarily worse than the optimal one.
\begin{theorem}
\label{thm:stochastic-eq}
Assuming the followers stochastically reach a Nash equilibrium drawn from a distribution over all equilibria, solving a Stackelberg game under the assumptions of optimistic or pessimistic equilibrium can give the leader expected payoff that is arbitrarily worse than the optimal one.
\end{theorem}
\begin{proof}
We consider a Stackelberg game with one leader and two followers (row and column player) with no constraint. The leader can choose 3 different strategies, each corresponding to a payoff matrix in Figure~\ref{fig:strat1}--\ref{fig:strat3}, where both followers receive the same payoff in the entry when they choose the corresponding row and column. In each payoff matrix, there are three pure Nash equilibria; we assume the followers reach any of them uniformly at random. After the followers reach a Nash equilibrium, the leader receives the corresponding entry in the payoff matrix in Figure~\ref{fig:strat-leader}.

Under the optimistic assumption, the leader would choose strategy~$1$, expecting followers to break the tie in favor of the leader, yielding payoff $C$. Instead, the three followers select a Nash equilibria uniformly at random, yielding expected payoff $\frac{C+0-C}{3} = 0$. Under the pessimistic assumption, the leader chooses strategy~$2$, anticipating and receiving an expected payoff of zero. Under the correct stochastic assumption, she chooses strategy~$3$ with expected payoff $\frac{C - \epsilon + C - \epsilon - \epsilon}{3} = \frac{2}{3} C - \epsilon \gg 0$, which can be arbitrarily higher than the optimistic or pessimistic payoff when $C \to \infty$.
\end{proof}

Theorem~\ref{thm:stochastic-eq} justifies why we need to work on stochastic equilibrium when the equilibrium is drawn stochastically in Definition~\ref{def:stochastic-problem}.
In the following section, we show how to apply gradient descent to optimize the leader's payoff by differentiating through followers' equilibria with a stochastic oracle.



\subsection{Equilibrium Flow and Unbiased Gradient Estimate}
Our goal is to compute the gradient of the objective in Equation~\ref{eqn:expected-utility}: $\frac{d}{d \environment} \mathop{\mathbb{E}}_{\boldsymbol\decision^* \sim \mathcal{O}(\environment)} f(\boldsymbol\decision^*, \environment)$. However, since the distribution of the oracle $\mathcal{O}(\environment)$ can also depend on $\environment$, we cannot easily exchange the gradient operator into the expectation.

To address the dependency of the oracle $\mathcal{O}(\environment)$ on $\environment$, we use $p(\boldsymbol\decision, \environment)$ to represent the probability density function of the oracle $\boldsymbol\decision \sim \mathcal{O}(\environment)$ for every $\environment$. We want to study how the oracle distribution changes as the leader's strategy $\environment$ changes, which we denote by {\it equilibrium flow} as defined by the following partial differential equation:

\begin{definition}[Equilibrium Flow]\label{def:conservation-law}
We call $v(\boldsymbol\decision, \environment)$ the equilibrium flow of the oracle $\mathcal{O}$ with probability density function $p(\boldsymbol\decision, \environment)$ if $v(\boldsymbol\decision, \environment)$ satisfies the following equation:
\begin{align}\label{eqn:conservation-law}
&\frac{\partial}{\partial \environment} p(\boldsymbol\decision, \environment) = - \nabla_{\boldsymbol\decision} \cdot (p(\boldsymbol\decision, \environment) v(\boldsymbol\decision, \environment)).
\end{align}
\end{definition}
This differential equation is similar to many differential equations of various conservation laws, where $v(\boldsymbol\decision, \environment)$ serves as a velocity term to characterize the movement of equilibria.
In the following theorem, we use the equilibrium flow $v(\boldsymbol\decision, \environment)$ to address the dependency of $\mathcal{O}(\environment)$ on $\environment$.

\begin{restatable}[]{theorem}{derivativeExpectation}\label{thm:expectation-of-derivative}
If $v(\boldsymbol\decision^*, \environment)$ is the equilibrium flow of the stochastic equilibrium oracle $\mathcal{O}(\environment)$, we have:
\begin{align}
    & \frac{d}{d \environment} \mathop{\mathbb{E}}\nolimits_{\boldsymbol\decision^* \sim \mathcal{O}(\environment)} f(\boldsymbol\decision^*, \environment) \nonumber \\
    = & \mathop{\mathbb{E}}\nolimits_{\boldsymbol\decision^* \sim \mathcal{O}(\environment)} \left[ f_\environment(\boldsymbol\decision^*, \environment) + f_{\boldsymbol\decision}(\boldsymbol\decision^*, \environment) \cdot v(\boldsymbol\decision^*, \environment) \right]. \label{eqn:equilibrium-sampling}
\end{align}
\end{restatable}
\begin{proof}[Proof sketch]
To compute the derivative on the left-hand side, we can expand the expectation by:
\begin{align}
    & \frac{d}{d \environment} \mathop{\mathbb{E}}\nolimits_{\boldsymbol\decision^* \sim \mathcal{O}(\environment)} f(\boldsymbol\decision^*, \environment) = \frac{d}{d \environment}  \int f(\boldsymbol\decision, \environment) p(\boldsymbol\decision, \environment) d \boldsymbol\decision 
    \nonumber \\ 
    & = \int p(\boldsymbol\decision, \environment) \frac{\partial}{\partial \environment} f(\boldsymbol\decision, \environment) \! + \! f(\boldsymbol\decision, \environment) \frac{\partial}{\partial \environment} p(\boldsymbol\decision, \environment) d \boldsymbol\decision \nonumber \\
    & = \mathop{\mathbb{E}}\nolimits_{\boldsymbol\decision^* \sim \mathcal{O}(\environment)} f_\environment(\boldsymbol\decision^*, \environment) + \int f(\boldsymbol\decision, \environment) \frac{\partial}{\partial \environment} p(\boldsymbol\decision, \environment) d \boldsymbol\decision \label{eqn:differentiate-through-multiplication}.
\end{align}

We substitute the term $\frac{\partial}{\partial \environment} p = - \nabla_{\boldsymbol\decision} \cdot (p \cdot v)$ by the definition of equilibrium flow, and apply integration by parts and Stokes' theorem\footnote{The analysis of integration by parts and Stokes' theorem applies to both Riemann and Lebesgue integral. Lebesgue integral is needed when the set of equilibria forms a measure-zero set.} to the right-hand side of Equation~\ref{eqn:differentiate-through-multiplication} to get Equation~\ref{eqn:equilibrium-sampling}. More details can be found in the appendix.
\end{proof}
Theorem~\ref{thm:expectation-of-derivative} extends the derivative of Nash equilibrium to the case of stochastic equilibrium randomly drawn from multiple equilibria. Specifically, Equation~\ref{eqn:differentiate-through-multiplication} offers an efficient unbiased gradient estimate by sampling an equilibrium from the stochastic oracle to compute the right-hand side of Equation~\ref{eqn:differentiate-through-multiplication}.
Theorem~\ref{thm:expectation-of-derivative} also matches to Equation~\ref{eqn:chain-rule}, where the role of equilibrium flow $v(\boldsymbol\decision^*, \environment)$ coincides with the role of $\frac{d \boldsymbol\decision^*}{d \environment}$ in Equation~\ref{eqn:chain-rule}.

\subsection{How to Determine Equilibrium Flow}
The only remaining question is how to determine the equilibrium flow. Given the leader's strategy $\environment$, there are two types of equilibria: (i)~isolated equilibria and (ii)~non-isolated equilibria. We first show that the solution to Equation~\ref{eqn:kkt_derivative_closed_form} matches the equilibrium flow for every equilibrium in case~(i) when the probability of sampling the equilibrium is locally fixed.


\begin{restatable}[]{theorem}{sufficientConditions}\label{thm:sufficient-conditions}
Given the leader's strategy~$\environment$ and a sampled equilibrium~$\boldsymbol\decision$, if (1) the KKT matrix at $(\boldsymbol\decision, \environment)$ is invertible and (2) $\boldsymbol\decision$ is sampled with a fixed probability locally, the solution to Equation~\ref{eqn:kkt_derivative_closed_form} is a homogeneous solution to Equation~\ref{eqn:conservation-law} and matches the equilibrium flow $v(\environment, \boldsymbol\decision)$ locally.
\end{restatable}
Theorem~\ref{thm:sufficient-conditions} ensures that when the sampled equilibrium behaves like a unique equilibrium locally, the solution to Equation~\ref{eqn:kkt_derivative_closed_form} matches the equilibrium flow of the sampled equilibrium.
In particular, Theorem~\ref{thm:sufficient-conditions} does not require all equilibria are isolated; it works as long as the sampled equilibrium satisfies the sufficient conditions. In contrast, the study in multiple equilibria requires global isolation for the analysis to work.
Together with Theorem~\ref{thm:expectation-of-derivative}, we can use the solution to Equation~\ref{eqn:kkt_derivative_closed_form} as an unbiased equilibrium gradient estimate and run stochastic gradient descent accordingly.


Lastly, when the sufficient conditions in Theorem~\ref{thm:sufficient-conditions} are not satisfied, e.g., the KKT matrix becomes singular for any non-isolated equilibrium, the solution to Equation~\ref{eqn:kkt_derivative_closed_form} does not match the equilibrium flow $v(\boldsymbol\decision, \environment)$. In this case, to compute the equilibrium flow correctly, we rely on solving the partial differential equation in Equation~\ref{eqn:conservation-law}.
If the probability density function $p(\boldsymbol\decision, \environment)$ is explicitly given, we can directly solve Equation~\ref{eqn:conservation-law} to derive the equilibrium flow.
If the probability density function $p(\boldsymbol\decision, \environment)$ is not given, we can use the empirical equilibrium distribution $p'(\boldsymbol\decision, \environment)$ constructed from the historical equilibrium samples of the oracle instead.

In practice, we hypothesize that even if the equilibria are not isolated and the corresponding KKT matrices are singular, solving degenerated version of Equation~\ref{eqn:kkt_derivative_closed_form} still serves as a good approximation to the equilibrium flow.
Therefore, we still use the solution to Equation~\ref{eqn:kkt_derivative_closed_form} as an approximate of the equilibrium flow in the following sections and algorithms.

\section{Gradient-Based Algorithm and Augmented Lagrangian Method}
To solve both the optimization problems in Definition~\ref{def:optimization-problem} and Definition~\ref{def:stochastic-problem}, we implement our algorithm with (i)~stochastic gradient descent with unbiased gradient access, and (ii)~augmented Lagrangian method to handle the equilibrium-dependent constraints.
We use the relaxation algorithm~\cite{uryas1994relaxation} as our equilibrium oracle~$\mathcal{O}$. The relaxation algorithm is a classic equilibrium finding algorithm that iteratively updates agents' strategies by best responding to other agents' strategies until convergence with guarantees~\cite{krawczyk2000relaxation}.


Since the leader's strategy~$\environment$ is constrained by the followers' response, we adopt an augmented Lagrangian method~\cite{bertsekas2014constrained} to convert the constrained problem to an unconstrained one with a Lagrangian objective.
We introduce a slack variable $\boldsymbol s \geq \boldsymbol 0$ to convert inequality constraints into equality constraints $\mathop{\mathbb{E}}_{\boldsymbol\decision^* \sim \mathcal{O}(\environment)} g(\boldsymbol\decision^*, \environment) + \boldsymbol s = \boldsymbol 0$.
Thus, the penalized Lagrangian can be written as:
\begin{align}\label{eqn:lagrangian}
    \mathcal{L}(\environment, \! \boldsymbol s; \! \boldsymbol\lambda) \! & = \! - \mathop{\mathbb{E}}\nolimits_{\boldsymbol\decision^* \! \sim \! \mathcal{O}(\environment)} \! f(\boldsymbol\decision^* \!,\! \environment) + \boldsymbol\lambda^\top \! (\mathop{\mathbb{E}}\nolimits_{\boldsymbol\decision^* \! \sim \! \mathcal{O}(\environment)} g(\boldsymbol\decision^* \!, \! \environment) \! + \! \boldsymbol s) \! \nonumber \\
    & \quad \quad + \frac{\mu}{2} \norm{\mathop{\mathbb{E}}\nolimits_{\boldsymbol\decision^* \! \sim \! \mathcal{O}(\environment)} g(\boldsymbol\decision^* \! , \! \environment) \! + \! \boldsymbol s}^2.
\end{align}
We run gradient descent on the minimization problem of the penalized Lagrangian $\mathcal{L}(\environment, \boldsymbol s; \boldsymbol\lambda)$ and update the Lagrangian multipliers $\boldsymbol\lambda$ every fixed number of iterations, 
as described in Algorithm~\ref{alg:constrained-differentiable-equilibrium}.
The stochastic Stackelberg problem with multiple followers can be solved by running stochastic gradient descent with augmented Lagrangian methods, where Theorem~\ref{thm:expectation-of-derivative} ensures the unbiasedness of the stochastic gradient estimate under the conditions in Theorem~\ref{thm:sufficient-conditions}.

\begin{algorithm}[t]
\caption{Augmented Lagrangian Method}
\label{alg:constrained-differentiable-equilibrium}
\textbf{Initialization:} $\environment = \environment_{\text{init}}$, learning rate $\gamma$, multipliers $\boldsymbol\lambda = \boldsymbol\lambda_0$, slack variable $\boldsymbol s \geq \boldsymbol 0, K=100$ \\
\For{\text{iteration in} $\{ 1,2,\dots \}$} {
Define the objective to be Lagrangian $\mathcal{L}(\environment, \boldsymbol s; \boldsymbol\lambda)$ defined in Equation~\ref{eqn:lagrangian} \\
Compute a sampled gradient of $\mathcal{L}$ by sampling $\boldsymbol\decision^* \sim \mathcal{O}(\environment)$. Compute $\frac{d \boldsymbol\decision^*}{d \environment}$ by Equation~\ref{eqn:kkt_derivative_closed_form} \\
Update $\environment = \environment - \gamma (\frac{\partial \mathcal{L}}{\partial \environment} + \frac{\partial \mathcal{L}}{\partial \boldsymbol\decision^*} \frac{d \boldsymbol\decision^*}{d \environment})$, $\boldsymbol s = \max\{\boldsymbol s - \gamma \frac{\partial \mathcal{L}}{\partial \boldsymbol s}, \boldsymbol 0\}$ \\
\If{iteration is a multiple of $K$}{
Update $\boldsymbol\lambda = \boldsymbol\lambda - \mu (g(\boldsymbol\decision^*, \environment) + \boldsymbol s)$
}
}
\textbf{Return:} leader's strategy $\environment$ 
\end{algorithm}


\section{Example Applications}\label{sec:example}
We briefly describe three different Stackelberg games with one leader and multiple self-interested followers.
Specifically, normal-form games with risk penalty has a unique Nash equilibrium, while other examples can have multiple.

\subsection{Coordination in Normal-Form Games}\label{sec:qre}
A normal-form game (NFG) is composed of $n$ follower players each with a payoff matrix $U_{i} \! \in \! \R^{m_1 \times \dots \times m_n}$ for all $i \! \in \! [n]$, where the $i$-th player has $m_i$ available pure strategies.
The set of all feasible mixed strategies of player~$i$ is $\decision_i \in \decisionset_i = \{ \decision \in [0,1]^{m_i} \mid \boldsymbol 1^\top \decision = 1 \}$.
On the other hand, the leader can offer non-negative subsidies $\environment_i \in \R_{\geq 0}^{m_1 \times \dots \times m_n}$ to each player~$i$ to reward specific combinations of pure strategies. The subsidy scheme is used to control the payoff matrix and incentivize the players to change their strategies.

Once the subsidy scheme $\environment$ is determined, each player $i$ chooses a strategy $\decision_i$ and receives the expected payoff $U_i(\boldsymbol\decision)$ and subsidy $\environment_i(\boldsymbol\decision)$, subtracting a penalty term $H(\decision_i) = \sum_j \decision_{ij} \log \decision_{ij}$, the Gibbs entropy of the chosen strategy $\boldsymbol\decision_i$ to represent the risk aversion of player $i$.
Since the followers' objectives are concave, the risk aversion model yields a unique Nash equilibrium, which is known to be quantal response equilibrium (QRE)~\cite{mckelvey1995quantal,ling2018game}.
Lastly, the leader's payoff is given by the social welfare across all players, which is the summation of the expected payoffs without subsidies: $\sum\nolimits_{i \in [n]} U_{i}(\boldsymbol\decision)$. The subsidy scheme is subject to a budget constraint $B$ on the total subsidy paid to all players.



\subsection{Security Games with Multiple Defenders}\label{sec:ssg}
Stackelberg security games (SSGs) model a defender protecting a set of targets $T$ from being attacked.
We consider a scenario with a leader coordinator and $n$ non-cooperative follower defenders each patrolling a subset $T_i \subseteq T$ of the targets~\cite{gan2018stackelberg}.
Each defender~$i$ can determine the patrol effort spent on protecting the designated targets. We use $0 \leq \decision_{i,t} \leq 1$ to denote the effort spent on target $t \in T_i$ and the total effort is upper bounded by $b_i$. 
Defender~$i$ only receives a penalty $U_{i,t} < 0$ when target $t \in T_{i}$ in her protected region is attacked but unprotected by all defenders, and $0$ otherwise.

Because the defenders are independent, the patrol strategies $\boldsymbol\decision$ can overlap, leading to a multiplicative unprotected probability $\prod\nolimits_{i} (1 - \decision_{i,t})$ of target~$t$. Given the unprotected probabilities, attacks occur under a distribution $\boldsymbol p \in \R^{|T|}$, where the distribution $\boldsymbol p$ is a function of the unprotected probabilities defined by a quantal response model.
To encourage collaboration, the leader coordinator can selectively provide reimbursement $\environment_{i,t} \geq 0$ to alleviate defender~$i$'s loss when target~$t$ is attacked but unprotected, which encourages the defender to focus on protecting specific regions, reducing wasted effort on overlapping patrols. The leader's goal is to protect all targets, where the leader's objective is the total return across over all targets $\sum\nolimits_{t \in T} U_{t} p_t \prod\nolimits_{i} (1 - \decision_{i,t})$.
Lastly, the reimbursement scheme $\environment$ must satisfy a budget constraint $B$ on the total paid reimbursement.



\subsection{Cyber Insurance Games With Multiple Customers}\label{sec:cyber}
We adopt the cyber insurance model proposed by Naghizadeh et al.~\shortcite{naghizadeh2014voluntary} and Johnson et al.~\shortcite{johnson2011security} to study how agents in an interconnected cyber security network make decisions, where agents' decisions jointly affect each other's risk. There are $n$ agents (followers) facing malicious cyberattacks. Each agent $i$ can deploy an effort of protection $\decision_i \in \R_{\geq 0}$ to his computer system, where investing in protection incurs a linear cost $c_i \decision_i$. Given the efforts $\boldsymbol\decision$ spent by all the agents, the joint protection of agent $i$ is $\sum\nolimits_{j=1}^n w_{ij} \decision_j$ with an interconnected effect parameterized by weights $W = \{w_{ij}\}_{i, j \in [n]}$. The probability of being attacked is modeled by $\sigma(- \sum\nolimits_{j=1}^n w_{ij} \decision_j + L_i)$, where $\sigma$ is the sigmoid function and $L_i$ refers to the value of agent $i$.

The Stackelberg leader is an external insurer who can customize insurance plans to influence agents' protection decisions. The leader can set an insurance plan  $\environment = \{I_i, \rho_i\}_{i \in [n]}$ to agent~$i$, where $\rho_i$ is the premium paid by agent~$i$ to receive compensation $I_i$ when attacked. Under the insurance plans and the interconnected effect, agents selfishly determine their effort spent on the protection $\boldsymbol\decision$ to maximize their payoff. On the other hand, the leader's objective is the total premium subtracting the compensation paid, while the constraints on the feasible insurance plans are the individual rationality of each customer, i.e., the compensation and premium must incentivize agents to purchase the insurance plan by making the payoff with insurance no worse than the payoff without. These constraints restrict the premium and compensation offered by the insurer.




\begin{figure*}[t]
    \centering
    \begin{subfigure}{0.325\linewidth}
        \centering
        \includegraphics[width=\textwidth]{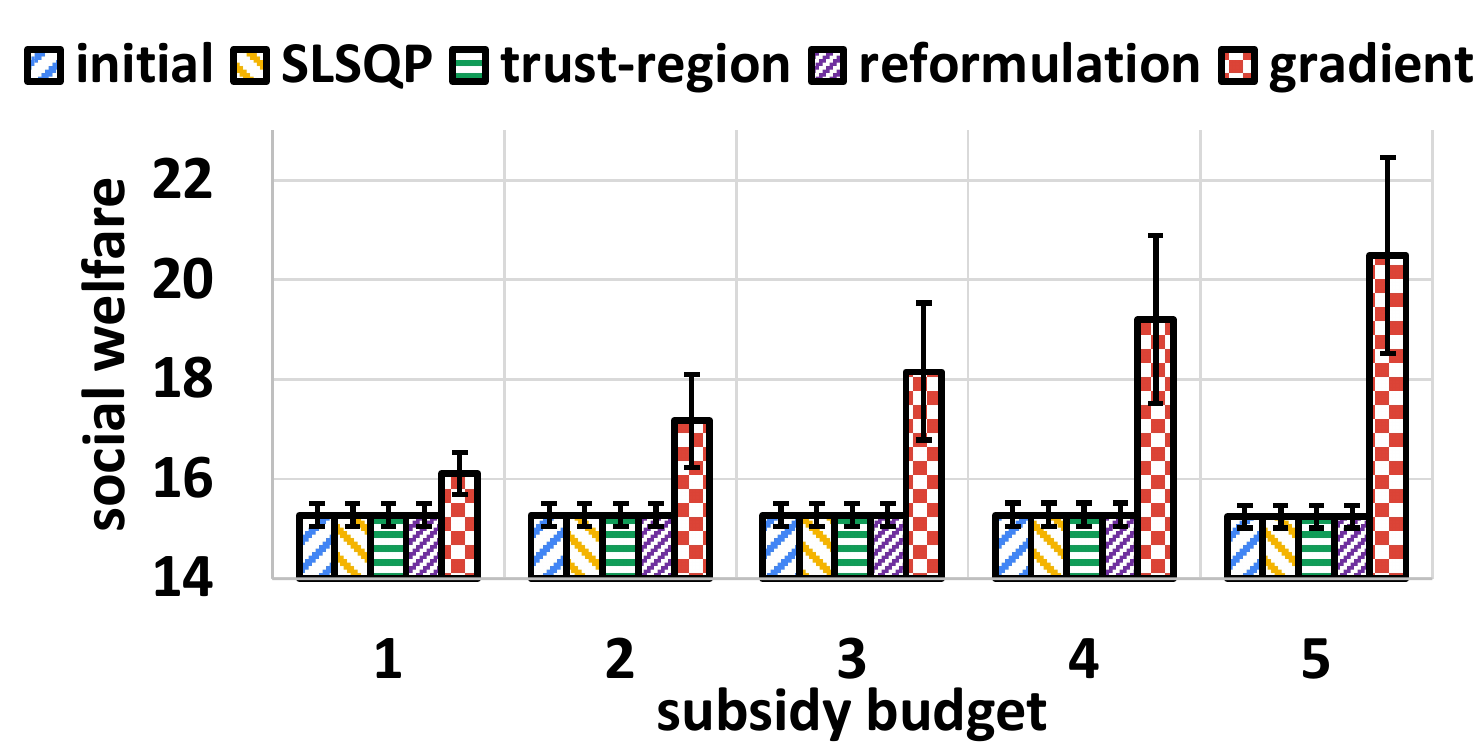}
        \caption{Normal-form games}
        \label{fig:qre}
    \end{subfigure}
    \hfill
    \begin{subfigure}{0.325\linewidth}
        \centering
        \includegraphics[width=\textwidth]{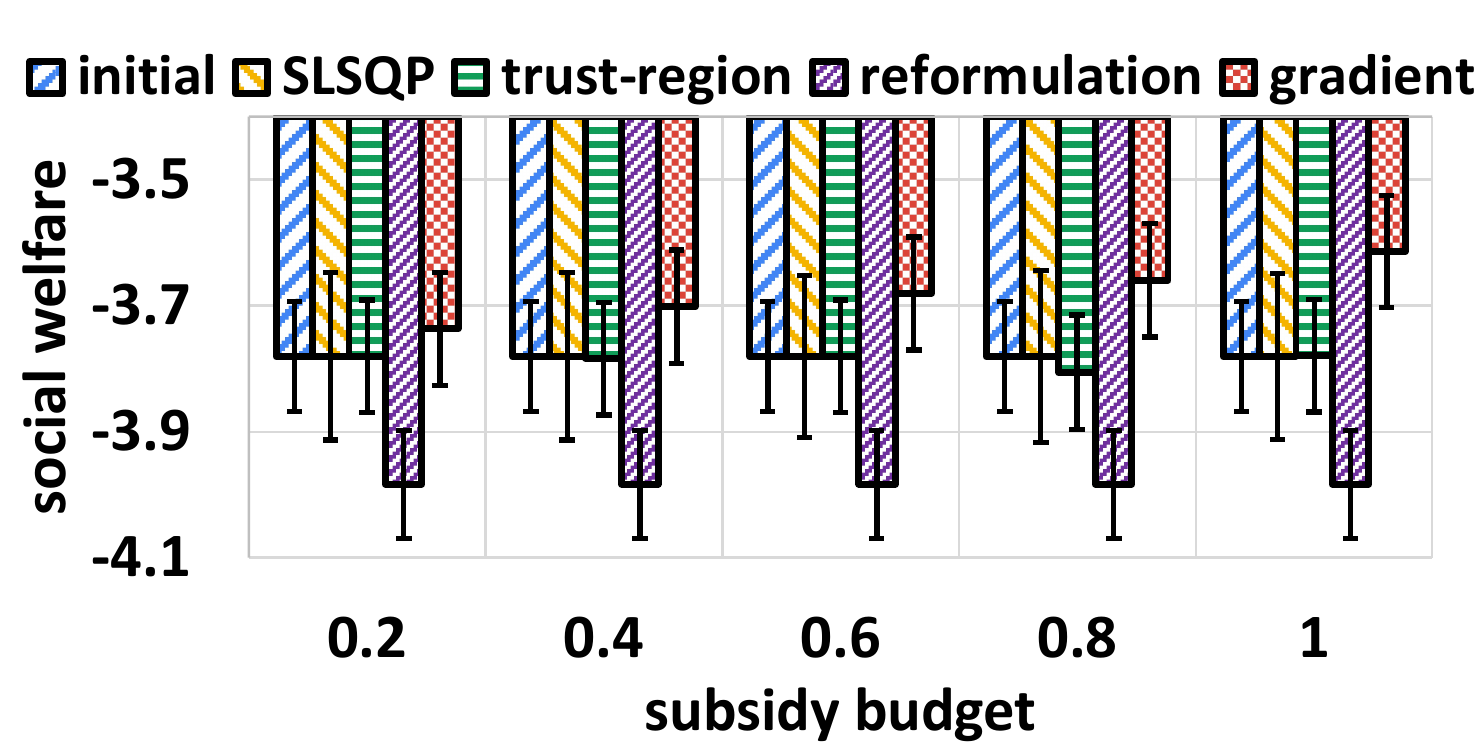}
        \caption{Multi-defender SSGs}
        \label{fig:ssg}
    \end{subfigure}
    \hfill
    \begin{subfigure}{0.325\linewidth}
        \centering
        \includegraphics[width=\linewidth]{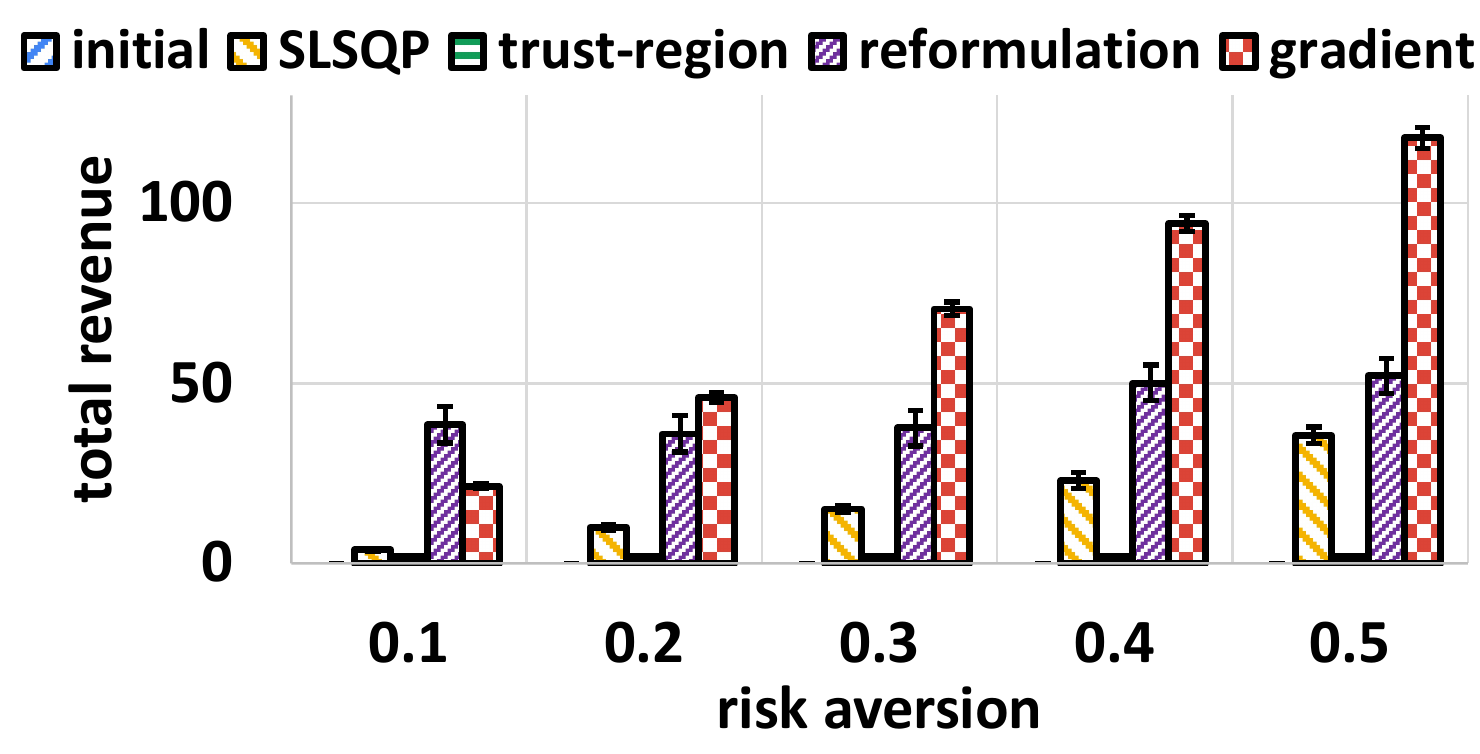}
        \vfill
        \caption{Cyber insurance games}
        \label{fig:cyber}
    \end{subfigure}
    \caption{We plot the solution quality of the Stackelberg problems with multiple followers. In all three domains, our gradient-based method achieves significantly higher objective than all other approaches. In NFGs and SSGs, the baselines cannot meaningfully improve upon the default strategy of the leader's initialization due to the high dimensionality of the parameter $\environment$; in cyber insurance games, SLSQP and reformulation both make some progress but still mostly with lower utility.}
    \label{fig:quality}
\end{figure*}
\begin{figure*}[t]
    \centering
    \begin{subfigure}{0.325\linewidth}
        \centering
        \includegraphics[width=\textwidth]{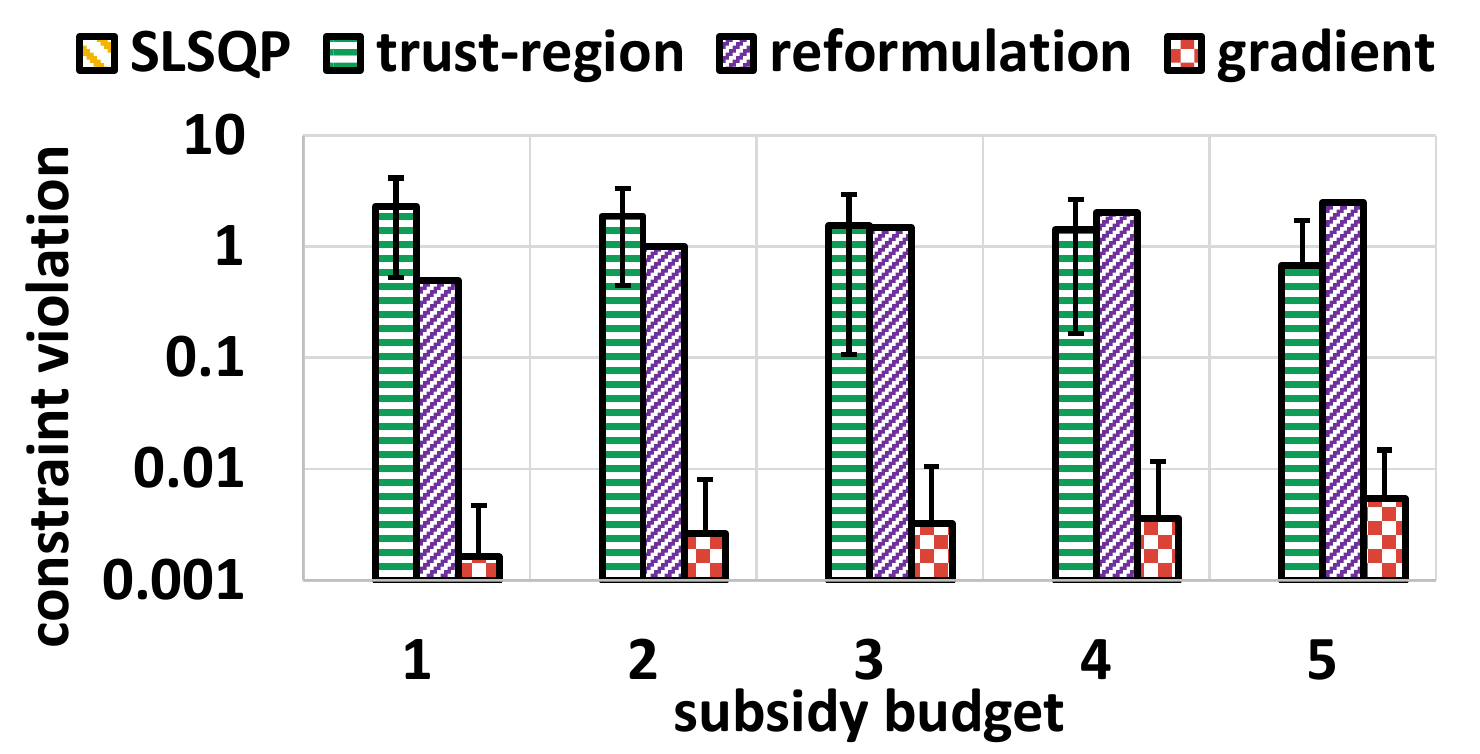}
        \caption{Normal-form games}
        \label{fig:qre-violation}
    \end{subfigure}
    \hfill
    \begin{subfigure}{0.325\linewidth}
        \centering
        \includegraphics[width=\textwidth]{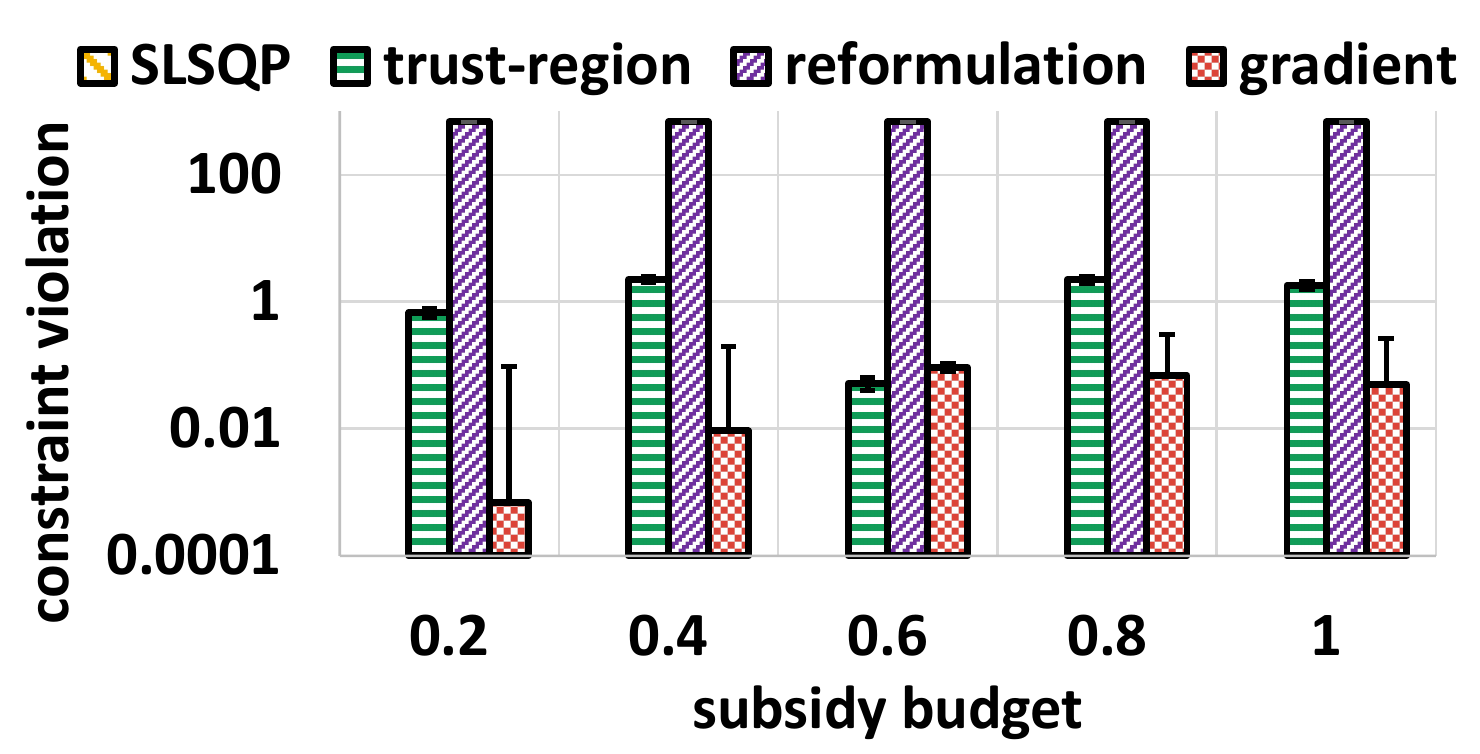}
        \caption{Multi-defender SSGs}
        \label{fig:ssg-violation}
    \end{subfigure}
    \hfill
    \begin{subfigure}{0.325\linewidth}
        \centering
        \includegraphics[width=\textwidth]{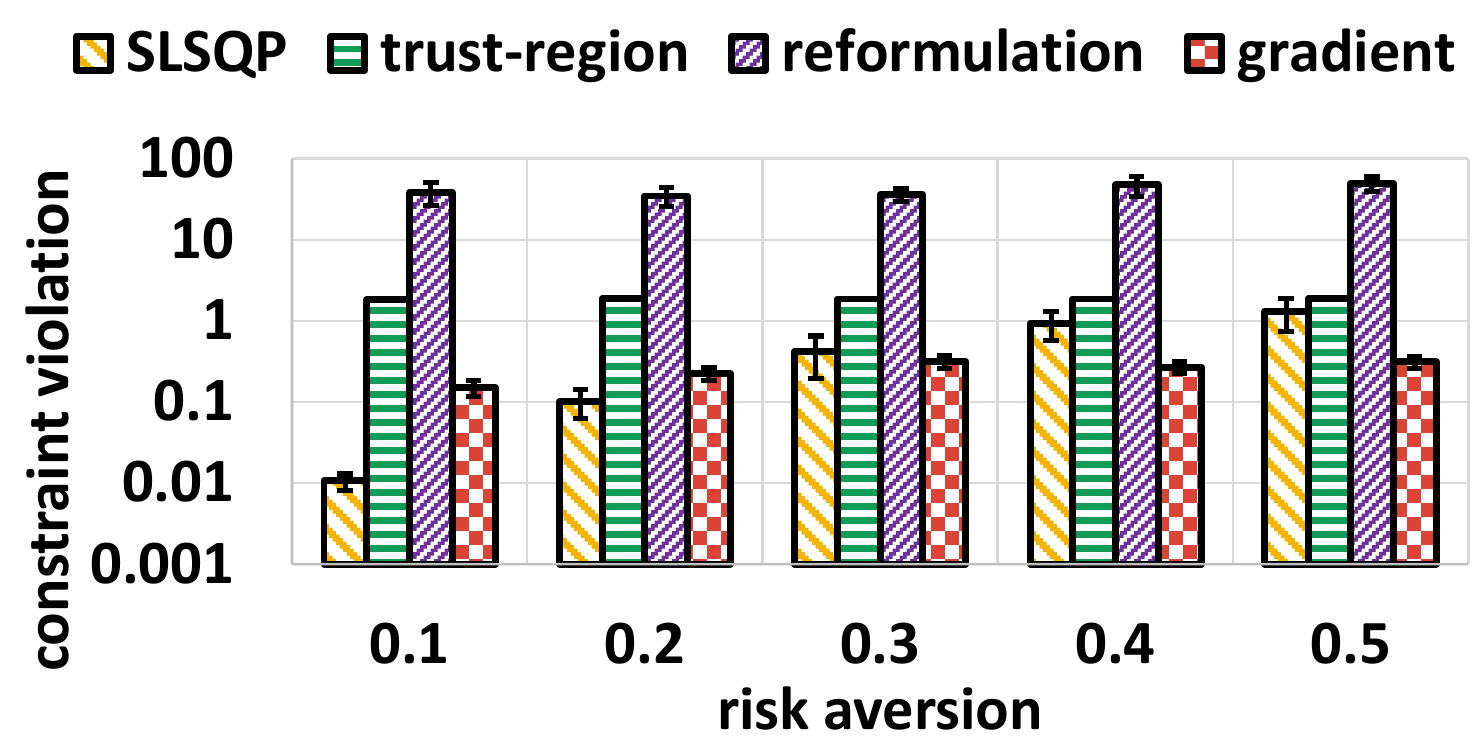}
        \caption{Cyber insurance games}
        \label{fig:cyber-violation}
    \end{subfigure}
    \caption{We plot the average budget constraint violation. Our gradient-based approach maintains low violation across all settings. SLSQP produces no violation in the first two domains because it fails to provide any meaningful improvement against the leader's initialization. Other baselines violate constraints more (often by orders of magnitude) despite less performance improvement.}
    \label{fig:constraint-violation}
\end{figure*}

\section{Experiments and Discussion}
We compare our gradient-based Algorithm~\ref{alg:constrained-differentiable-equilibrium} ({\bf gradient}) against various baselines in the three settings described above.
In each experiment, we execute $30$ independent runs ($100$ runs for SSGs) under different randomly generated instances. We run Algorithm~\ref{alg:constrained-differentiable-equilibrium} with learning rate $\gamma = 0.01$ for 5,000 gradient steps and update the Lagrange multipliers every $K \! = \! 100$ iterations. Our gradient-based method completes in about an hour across all settings---refer to the appendix for more details. 

\paragraph{Baselines}
We compare against several baselines that can solve the stochastic Stackelberg problem with multiple followers with equilibrium-dependent objective and constraints. In particular, given the non-convexity of agents' objective functions, SSGs and cyber insurance games can have multiple, stochastic equilibria. Our first baseline is the leader's {\bf initial} strategy $\environment_0$, which is a naive all-zero strategy in all three settings.
Blackbox optimization baselines include sequential least squares programming ({\bf SLSQP})~\cite{kraft1988software} and the {\bf trust-region} method~\cite{conn2000trust}, where the equilibrium encoded in the optimization problem is treated as a blackbox that needs to be repeatedly queried.
{\bf Reformulation}-based algorithm ~\cite{basilico2020bilevel,aussel2020trilevel} is the state-of-the-art method to solve Stackelberg games with multiple followers. This approach reformulates the followers' equilibrium conditions into non-linear complementary constraints as a mathematical program with equilibrium constraints~\cite{luo1996mathematical}, then solves the problem using branch-and-bound and mixed integer non-linear programming (we use a commercial solver, Knitro~\cite{nocedal2006knitro}). The reformulation-based approach cannot handle arbitrary stochastic equilibria but can handle optimistic or pessimistic equilibria. We implement the optimistic version of the reformulation as our baseline, which could potentially suffer from a performance drop as exemplified in Theorem~\ref{thm:stochastic-eq}. 


\subsection{Solution Quality}
In Figure~\ref{fig:qre} and~\ref{fig:ssg}, we plot the leader's objective ($y$-axis) versus various budgets for the paid subsidy ($x$-axis). Figure~\ref{fig:cyber}, shows the total revenue to the insurer ($y$-axis) versus the risk aversion of agents ($x$-axis).
Denoting the number of agents by $n$ and the number of actions per agent by $m$, we have $n=3,5,10$ and $m=10,50,1$ in NFGs, SSGs, and cyber insurance games, respectively.

Our optimization baselines perform poorly in Figure~\ref{fig:qre} and~\ref{fig:ssg} due to the high dimensionality of the environment parameter $\environment$ in NFGs ($\dim(\environment) = n m^n$) and SSGs ($\dim(\environment) = nm$), respectively.
In Figure~\ref{fig:cyber}, the dimensionality of cyber insurance games ($\dim(\environment) = 2n$) is smaller, where we can see that SLSQP and reformulation-based approaches start making some progress, but still less than our gradient-based approach.
The main reason that blackbox methods do not work is due to the expensive computation of numerical gradient estimates. On the other hand, reformulation method fails to handle the mixed-integer non-linear programming problem reformulated from followers' best response and the constraints within a day.



\subsection{Constraint Violation}
In Figure \ref{fig:constraint-violation}, we provide the average constraint violation across different settings. Blackbox optimization algorithms either become stuck at the initial point due to the inexact numerical gradient estimate or create large constraint violations due to the complexity of equilibrium-dependent constraints. The reformulation approach also creates large constraint violations due to the difficulty of handling large number of non-convex followers' constraints under high-dimensional leader's strategy.
In comparison, our method can handle equilibrium-dependent constraints by using an augmented Lagrangian method with an ability to tighten the budget constraint violation under a tolerance as shown.
Although Figure~\ref{fig:constraint-violation} only plots the budget constraint violation, in our algorithm, we enforce that the equilibrium oracle runs until the equilibrium constraint violation is within a small tolerance $10^{-6}$, whereas other algorithms sometimes fail to satisfy such equilibrium constraints.

\section{Conclusion}
In this paper, we present a gradient-based approach to solve Stackelberg games with multiple followers by differentiating through followers' equilibrium to update the leader's strategy. Our approach generalizes to stochastic gradient descent when the equilibrium reached by followers is stochastically chosen from multiple equilibria. We establish the unbiasedness of the stochastic gradient by the equilibrium flow derived from a partial differential equation. 
To our knowledge, this work is the first to establish the unbiasedness of gradient computed from stochastic sample of multiple equilibria.
Empirically, we implement our gradient-based algorithm on three different examples, where our method outperforms existing optimization and reformulation baselines.


\section*{Acknowledgement}
This research was supported by MURI Grant Number W911NF-17-1-0370.
The computations in this paper were run on the FASRC Cannon cluster supported by the FAS Division of Science Research Computing Group at Harvard University.

\bibliography{reference}

\cleardoublepage
\newpage
\appendix
\section{Implementation Details}
We implement a differentiable PyTorch module to compute a sample of the followers' equilibria. The module takes the leader's strategy as input and outputs a Nash equilibrium computed in the forward pass using the relaxation algorithm. We use a random initialization to run the relaxation algorithm, which can reach to different equilibria depending on different initialization. Given the sampled equilibrium $\boldsymbol\decision^*$ computed in the forward pass, the backward pass is implemented by PyTorch autograd to compute all the second-order derivatives to express Equation~\ref{eqn:kkt_derivative_closed_form}.
The backward pass solves the linear system in Equation~\ref{eqn:kkt_derivative_closed_form} analytically to derive $\frac{d \boldsymbol\decision^*}{d \environment}$ as an approximate of the equilibrium flow.

This PyTorch module is used in all three examples in our experiment. The implementation is flexible as we just need to adjust the followers' objectives and constraints, the relaxation algorithm and the gradient computation all directly apply.

\section{Proofs of Theorem~\ref{thm:expectation-of-derivative} and Theorem~\ref{thm:sufficient-conditions}}\label{sec:missing-proofs}

\derivativeExpectation*
\begin{proof}
To compute the derivative on the left-hand side, we have to first expand the expectation because the equilibrium distribution is dependent on the environment parameter $\environment$:
\begin{align}
    & \frac{d}{d \environment} \mathop{\mathbb{E}}_{\boldsymbol\decision \sim \mathcal{O}(\environment)} f(\boldsymbol\decision, \environment) \nonumber \\ =& \frac{d}{d \environment}  \int_{\boldsymbol\decision \in \decisionset}   f(\boldsymbol\decision, \environment) p(\boldsymbol\decision, \environment) d \boldsymbol\decision \nonumber \\
    =& \int_{\boldsymbol\decision \in \decisionset} \left( p(\boldsymbol\decision, \environment) \frac{\partial}{\partial \environment} f(\boldsymbol\decision, \environment) + f(\boldsymbol\decision, \environment) \frac{\partial}{\partial \environment} p(\boldsymbol\decision, \environment) \right) d \boldsymbol\decision \nonumber \\
    =& \mathop{\mathbb{E}}_{\boldsymbol\decision \sim \mathcal{O}(\environment)} f_\environment(\boldsymbol\decision, \environment) + \int_{\boldsymbol\decision \in \decisionset} f(\boldsymbol\decision, \environment) \frac{\partial}{\partial \environment} p(\boldsymbol\decision, \environment) d \boldsymbol\decision \label{eqn:differentiate-through-multiplication-appendix}
\end{align}

We further define $\Phi(\boldsymbol\decision, \environment) = p(\boldsymbol\decision, \environment) v(\boldsymbol\decision, \environment)$. By the equilibrium flow definition in Equation~\ref{eqn:conservation-law}, we have
\begin{align}
\frac{\partial}{\partial \environment} p(\boldsymbol\decision, \environment) = - \nabla_{\boldsymbol\decision} \cdot \Phi(\boldsymbol\decision, \environment) \nonumber
\end{align}

Therefore, the later term in Equation~\ref{eqn:differentiate-through-multiplication-appendix} can be computed by integration by parts and Stokes' theorem:
\begin{align}
    &  \int_{\boldsymbol\decision \in \decisionset} f(\boldsymbol\decision, \environment) \frac{\partial}{\partial \environment} p(\boldsymbol\decision, \environment) d \boldsymbol\decision \nonumber \\
    =& - \int_{\decision \in \decisionset} f(\boldsymbol\decision, \environment) \nabla_{\boldsymbol\decision} \cdot \Phi(\boldsymbol\decision, \environment) d \boldsymbol\decision \nonumber \\
    =& - \int_{\decision \in \decisionset} \nabla_{\boldsymbol\decision} \cdot (f(\boldsymbol\decision, \environment) \Phi(\boldsymbol\decision, \environment)) d \boldsymbol\decision \nonumber \\
    & \quad \quad \quad \quad \quad \quad \quad \quad + \int_{\boldsymbol\decision \in \decisionset} f_{\boldsymbol\decision}(\boldsymbol\decision, \environment) \Phi(\boldsymbol\decision, \environment) d \boldsymbol\decision \nonumber \\
    =& - \oint_{\partial \decisionset} f(\boldsymbol\decision, \environment) \Phi(\boldsymbol\decision, \environment) dS + \int_{\boldsymbol\decision \in \decisionset} f_{\boldsymbol\decision}(\boldsymbol\decision, \environment) \Phi(\boldsymbol\decision, \environment) d \boldsymbol\decision \nonumber
\end{align}

Therefore, we have
\begin{align}
    & \frac{d}{d \environment} \mathop{\mathbb{E}}_{\boldsymbol\decision \sim \mathcal{O}(\environment)} f(\boldsymbol\decision, \environment) \nonumber \\
    =& \mathop{\mathbb{E}}_{\boldsymbol\decision \sim \mathcal{O}(\environment)} f_\environment(\boldsymbol\decision, \environment) + \int_{\boldsymbol\decision \in \decisionset} f(\boldsymbol\decision, \environment) \frac{\partial}{\partial \environment} p(\boldsymbol\decision, \environment) d \boldsymbol\decision \nonumber \\
    =& \mathop{\mathbb{E}}_{\boldsymbol\decision \sim \mathcal{O}(\environment)} f_\environment(\boldsymbol\decision, \environment) - \oint_{\partial \decisionset} f(\boldsymbol\decision, \environment) \Phi(\boldsymbol\decision, \environment) dS \nonumber \\
    & \quad \quad + \int_{\boldsymbol\decision \in \decisionset} f_{\boldsymbol\decision}(\boldsymbol\decision, \environment) \Phi(\boldsymbol\decision, \environment) d \boldsymbol\decision \nonumber \\
    =& \mathop{\mathbb{E}}_{\boldsymbol\decision \sim \mathcal{O}(\environment)} f_\environment(\boldsymbol\decision, \environment) - \oint_{\partial \decisionset} f(\boldsymbol\decision, \environment) p(\boldsymbol\decision, \environment) v(\boldsymbol\decision, \environment) dS \nonumber \\
    & \quad \quad + \int_{\boldsymbol\decision \in \decisionset} f_{\boldsymbol\decision}(\boldsymbol\decision, \environment) p(\boldsymbol\decision, \environment) v(\boldsymbol\decision, \environment) d \boldsymbol\decision \nonumber \\
    =& \mathop{\mathbb{E}}_{\boldsymbol\decision \sim \mathcal{O}(\environment)} \left[ f_\environment(\boldsymbol\decision, \environment) + f_{\boldsymbol\decision}(\boldsymbol\decision, \environment) v(\boldsymbol\decision, \environment) \right] \nonumber
\end{align}
where the term $\oint_{\partial \decisionset} f(\boldsymbol\decision, \environment) p(\boldsymbol\decision, \environment) v(\boldsymbol\decision, \environment) dS = 0$ because $p(\boldsymbol\decision, \environment) = 0$ at the boundary $\partial \decisionset$. This concludes the proof of Theorem~\ref{thm:expectation-of-derivative}.
\end{proof}

Notice that in the proof of Theorem~\ref{thm:expectation-of-derivative}, we only use integration by parts and Stokes' theorem, where both of them apply to Riemann integral and Lebesgue integral. Thus, the proof of Theorem~\ref{thm:expectation-of-derivative} also works for any measure zero jumps in the probability density function.

\sufficientConditions*
\begin{proof}
Since the KKT conditions hold for all equilibria, the given $\environment$ and $\boldsymbol\decision$ must satisfy $KKT(\boldsymbol\decision, \environment) = 0$. The KKT matrix in Equation~\ref{eqn:kkt_derivative_closed_form} is given by $\frac{\partial KKT}{\partial \boldsymbol\decision}$, the Jacobian matrix of the function $KKT(\boldsymbol\decision, \environment)$ with respect to $\boldsymbol\decision$. If the KKT matrix is invertible, by implicit function theorem, there exists an open set $U$ containing $\environment$ such that there exists a unique continuously differentiable function $h: U \rightarrow \decisionset$ such that $h(\environment) = \boldsymbol\decision$ and $KKT(h(\environment'), \environment') = 0$ for all $\environment' \in U$.
Moreover, the analysis in Equation~\ref{eqn:kkt_derivative_closed_form} applies, where $\frac{d h(\environment)}{d \environment} = \frac{d \boldsymbol\decision}{d \environment}$ matches the solution of Equation~\ref{eqn:kkt_derivative_closed_form}.

Lastly, the condition that the equilibrium $\boldsymbol\decision$ is sampled with a fixed probability density $c$ locally implies the corresponding probability density function must satisfy $p(\decision', \environment') = c \mathbf{1}_{\text{KKT}(\decision', \environment') = 0} = c \mathbf{1}_{\boldsymbol\decision' = h(\environment')}$ for all $\environment' \in U$ in an open set locally\footnote{We can choose the smaller subset $U$ such that both the implicit function theorem and the locally fixed probability $c$ both hold.}.

Now we can verify whether $p(\boldsymbol\decision', \environment')$ and $v(\boldsymbol\decision', \environment') = \frac{d h(\environment')}{d \environment}$ (independent of $\boldsymbol\decision'$) satisfy the partial differential equation of equilibrium flow as defined in Definition~\ref{def:conservation-law}.
We first compute the left-hand side of Equation~\ref{eqn:conservation-law} by:
\begin{align}
    \frac{\partial}{\partial \environment} p(\boldsymbol\decision', \environment') &= \frac{\partial}{\partial \environment} c \mathbf{1}_{\boldsymbol\decision' = h(\environment')} \nonumber \\ 
    &= c \delta_{\boldsymbol\decision' = h(\environment')} \frac{d h(\environment')}{d \environment} \label{eqn:differential-equation-lhs}
\end{align}
where Equation~\ref{eqn:differential-equation-lhs} is derived by fixing $\boldsymbol\decision'$, the derivative of a jump function $\mathbf{1}_{\boldsymbol\decision' = h(\environment')}$ is a Dirac delta function located at $\boldsymbol\decision' = h(\environment')$ multiplied by a Jacobian term $\frac{d h(\environment')}{d \environment}$.

We can also compute the right-hand side of Equation~\ref{eqn:conservation-law} by:
\begin{align}
    & \nabla_\decision \cdot (p(\boldsymbol\decision', \environment') v(\boldsymbol\decision', \environment')) \nonumber \\
    =& v(\boldsymbol\decision', \environment') \frac{\partial}{\partial \boldsymbol\decision} p(\boldsymbol\decision', \environment') + p(\boldsymbol\decision', \environment') \frac{\partial}{\partial \boldsymbol\decision} v(\boldsymbol\decision', \environment') \label{eqn:differential-equation-chain-rule} \\
    =& \frac{d h(\environment')}{d \environment} \frac{\partial}{\partial \boldsymbol\decision} c\boldsymbol1_{\boldsymbol\decision'=h(\environment')} \nonumber \\
    =& c \delta_{\boldsymbol\decision' = h(\environment')} \frac{d h(\environment')}{d \environment} \label{eqn:differential-equation-rhs}
\end{align}
where the second term in Equation~\ref{eqn:differential-equation-chain-rule} is $0$ because we define $v(\boldsymbol\decision', \environment') = \frac{d h(\environment')}{d \environment}$, which is independent of $\boldsymbol\decision'$.
Equation~\ref{eqn:differential-equation-rhs} is derived by fixing $\environment'$, the derivative of a jump function is a Dirac delta function located at $\boldsymbol\decision' = \environment'$.

The above calculation shows that Equation~\ref{eqn:differential-equation-lhs} is identical to Equation~\ref{eqn:differential-equation-rhs}, which implies the left-hand side and the right-hand side of Equation~\ref{eqn:conservation-law} are equal.
Therefore, we conclude that the choice of $v(\boldsymbol\decision', \environment') = \frac{d \boldsymbol\decision'}{d \environment} = \frac{d h(\environment')}{d \environment}$ is a homogeneous solution to differential equation in Equation~\ref{eqn:conservation-law} locally in $\environment' \in U$. 
By the definition of the equilibrium flow, $v(\boldsymbol\decision', \environment') = \frac{d \boldsymbol\decision'}{d \environment}$ is a solution to the equilibrium flow because we can subtract the homogeneous solution and define a new partial differential equation without region $U$ to compute the solution outside of $U$.
\end{proof}

\begin{figure*}[t]
    \centering
    \begin{subfigure}{0.32\linewidth}
        \includegraphics[width=\textwidth]{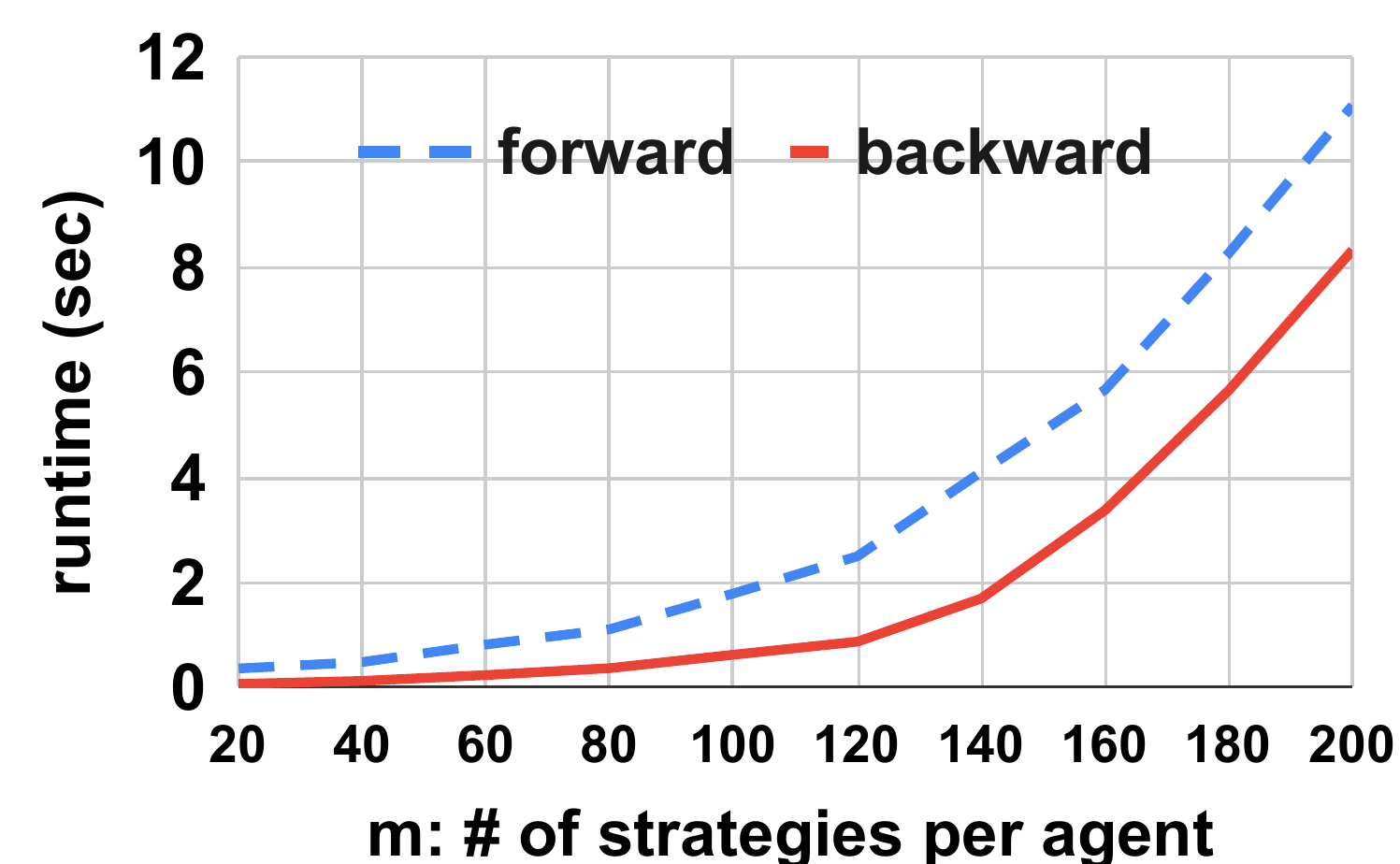}
        \caption{Normal-form games with $n=3$ followers and varied $m$ strategies per follower.}
        \centering
        \label{fig:qre-runtime}
    \end{subfigure}
    \hfill
    \begin{subfigure}{0.32\linewidth}
        \includegraphics[width=\textwidth]{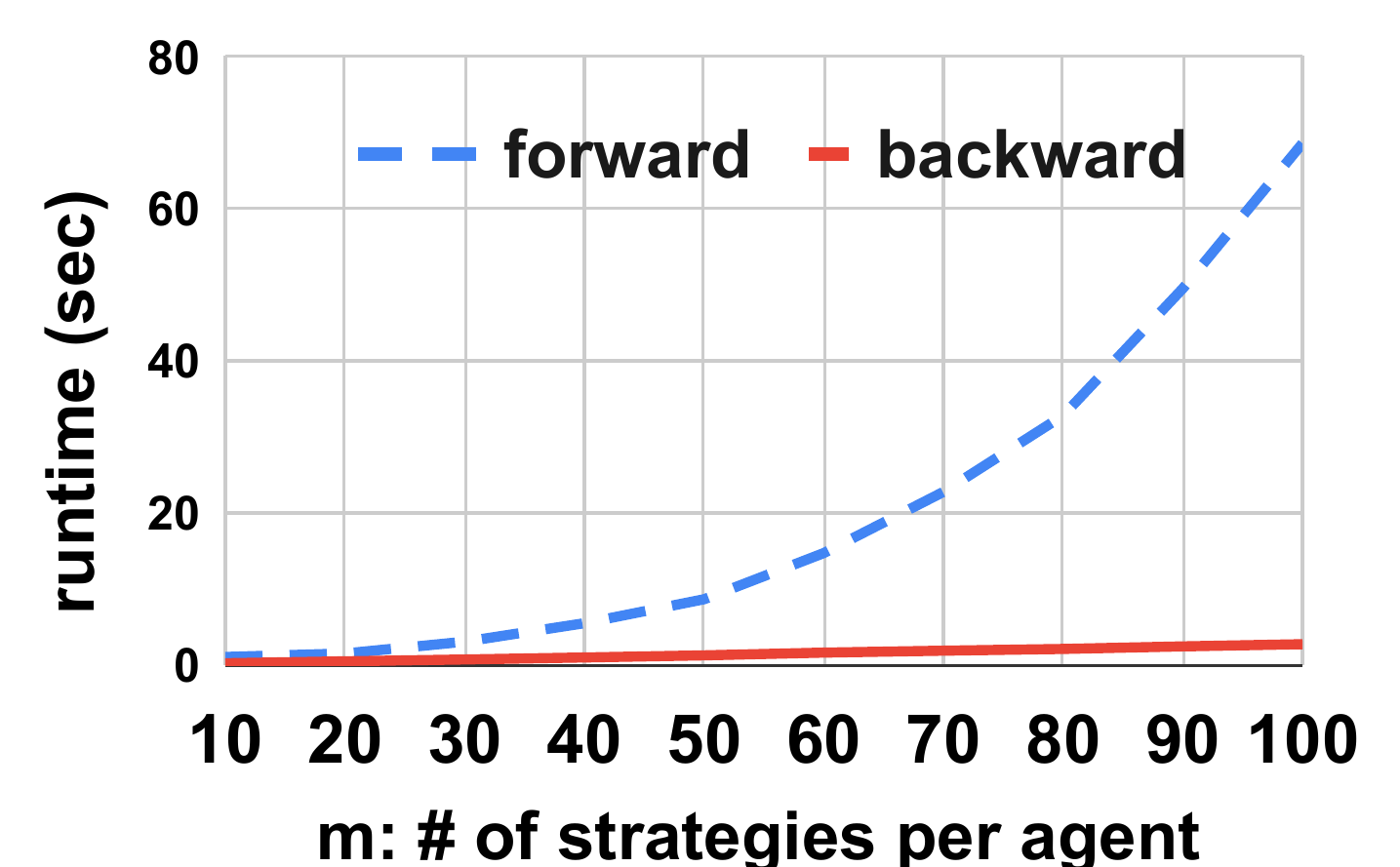}
        \caption{Stackelberg security games with $n=5$ followers and varied $m$ strategies per follower.}
        \centering
        \label{fig:ssg-runtime}
    \end{subfigure}
    \hfill
    \begin{subfigure}{0.32\linewidth}
        \includegraphics[width=\textwidth]{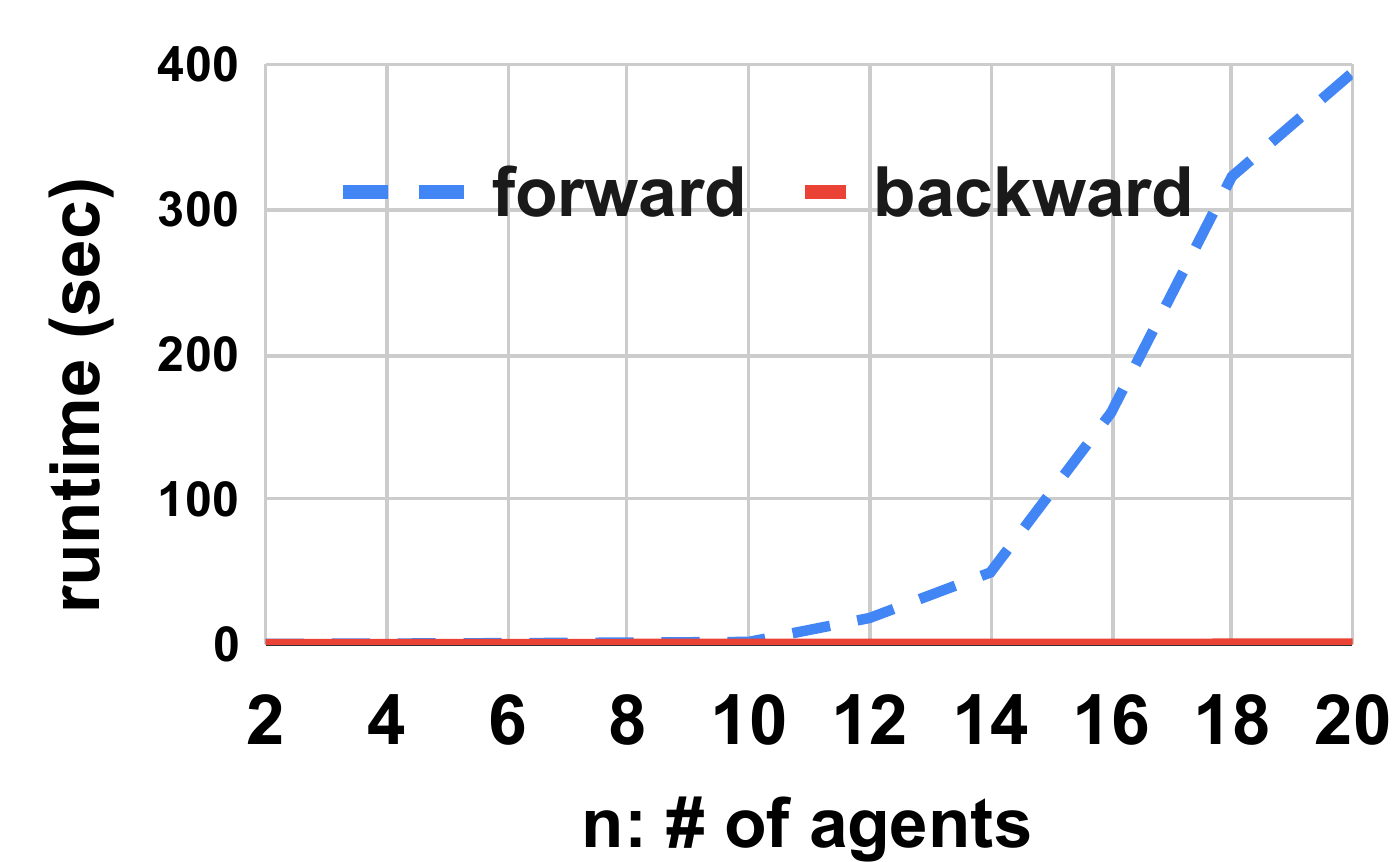}
        \caption{Cyber insurance games with $m=1$ dimensional investment decision and varied $n$.}
        \label{fig:cyber-runtime}
        \centering
    \end{subfigure}
    \caption{We compare the computation cost of equilibrium computation (forward) and the gradient access (backward) per iteration. Backward pass is cheaper than forward pass in all three domains. Gradient-based method runs a forward pass and a backward pass per iteration, while gradient-free method requires many forward passes to perform one step of local search.}
    \label{fig:runtime}
\end{figure*}

\section{Limitation of Theorem~\ref{thm:expectation-of-derivative} and Theorem~\ref{thm:sufficient-conditions}}\label{sec:limitations}
Although Theorem~\ref{thm:expectation-of-derivative} always holds, the main challenge preventing us from directly applying Theorem~\ref{thm:expectation-of-derivative} is that we do not know the equilibrium flow in advance.
Given the probability density function of the equilibrium oracle, we can compute the equilibrium flow by solving the partial differential equation in Equation~\ref{eqn:conservation-law}.
However, the probability density function is generally not given. 

Theorem~\ref{thm:sufficient-conditions} tells us that the derivative computed in Equation~\ref{eqn:kkt_derivative_closed_form} is exactly the equilibrium flow defined by the partial differential equation when the sampled equilibrium admits to an invertible KKT matrix and is locally sampled with a fixed probability.
That is to say, when these conditions hold, we can treat the equilibrium sampled from a distribution over multiple equilibria as a unique equilibrium to differentiate through as discussed in the section of unique Nash equilibrium.
These conditions are also satisfied when the sampled equilibrium is locally stable without any discontinuous jump, generalizing the differentiability of unique Nash equilibrium and globally isolated Nash equilibria to the case with only conditions on the sampled Nash equilibrium.

\section{Dimensionality and Computation Cost}

\subsection{Dimensionality of Control Parameters}
We discuss the solution quality attained and computation costs required by different optimization methods.
To understand the results, it is useful to compare the role and dimensionality of the environment parameter $\environment$ in each setting.
\begin{itemize}
    \item \emph{Normal-form games}: parameter $\environment$ corresponds to the non-negative subsidies provided to each follower for each entry of its payoff matrix. We have $\dim(\environment) = n \prod\nolimits_{i=1}^n m_i = n m^n$, where for simplicity we set $m_i = m$ for all $i$.
    \item \emph{Stackelberg security games}: parameter $\environment$ refers to the non-negative subsidies provided to each follower at each available target. Because each follower $i$ can only cover targets $T_i \subseteq T$, we have $\dim(\environment) = \sum\nolimits_{i=1}^n |T_i| = nm$, where we set $|T_i| = m$ for all $i$.
    \item \emph{Cyber insurance games}: each insurance plan is composed of a premium and a coverage amount. Therefore in total, $\dim(\environment) = 2n$, the smallest out of the three tasks.
\end{itemize}

\subsection{Computation Cost}
\label{sec:computation-cost}
In Figure~\ref{fig:runtime}, we compare the computation cost per iteration of equilibrium-finding oracle (forward) and the gradient oracle (backward). Due to the hardness of the Nash equilibrium-finding problem, no equilibrium oracle is likely to have polynomial-time complexity in the forward pass (computing an equilibrium). We instead focus more on the computation cost of the backward pass (differentiating through an equilibrium).

As we can see in Equation~\ref{eqn:kkt_derivative_closed_form}, the complexity of gradient computation is dominated by inverting the KKT matrix with size $L = O(nm)$ and the dimensionality of environment parameter $\environment$ since the matrix $\frac{d \boldsymbol\decision^*}{d \environment}$ is of size $L \times \dim(\environment)$. Therefore, the complexity of the backward pass is bounded above by $O(L^\alpha) + O(L^2 \dim(\environment)) = O(n^\alpha m^\alpha) + O(n^2 m^2 \dim(\environment))$ with $\alpha = 2.373$.
\begin{itemize}
    \item In Figure~\ref{fig:qre-runtime}, the complexity is given by $O(n^2 m^2 \dim(\environment)) = O(n^3 m^{n+2}) = O(m^5)$ where we set $n=3$ with varied $m$, number of actions per follower, shown in the $x$-axis.
    \item In Figure~\ref{fig:ssg-runtime}, the complexity is $O(n^2 m^2 \dim(\environment)) = O(m^3)$ with $n=5$ and varied $m$, number of actions per follower, shown in the $x$-axis.
    \item In Figure~\ref{fig:cyber-runtime}, the complexity is $O(n^2 m^2 \dim(\environment)) = O(n^3)$ with $m=1$ and varied number of followers $n$ shown in the $x$-axis. The runtime of the forward pass increases drastically, while the runtime of the backward pass remains polynomial.
\end{itemize}
In all three examples, the gradient computation (backward) has polynomial complexity and is faster than the equilibrium finding oracle (forward).
Numerical gradient estimation in gradient-free methods requires repeatedly accessing the forward pass, which can be even more expensive than our gradient computation.

\section{Optimization Reformulation of the Stackelberg Problems with Multiple Followers}
In this section, we describe how to reformulate the leader's optimization problem with multiple followers involved into an single-level optimization problem with stationary and complementarity constraints.
Notice that this reformulation requires the assumption that all followers break ties in favor of the leader, while our gradient-based method can deal with arbitrary oracle access not limited to any tie-breaking rules.

\subsection{Normal-Form Games with Risk Penalty}
In this example, the followers' objectives are defined by:
\begin{align}\label{eqn:qre-individual-payoff}
    f_i(\boldsymbol\decision, \environment) = U_{i}(\boldsymbol\decision) + \environment_i(\boldsymbol\decision) - H(\decision_i) / \lambda,
\end{align}
where $U_i$ is the given payoff matrix and $\environment_i$ is the subsidy provided by the leader. $H$ is the Gibbs entropy denoting the risk aversion penalty.

The leader's objective and the constraint are respectively defined by:
\begin{align}
    f(\boldsymbol\decision, \environment) &= \sum\nolimits_{i \in [n]} U_{i}(\boldsymbol\decision) \nonumber \\
    g(\boldsymbol\decision, \environment) &= \left( \sum\nolimits_{i \in [n]} \environment_i(\boldsymbol\decision) \right) - B \leq 0.  \nonumber
\end{align}

\paragraph{Bilevel optimization formulation}
we can write the followers' best response into the leader's optimization problem:
\begin{align*}
    \max_{\boldsymbol\environment} \quad & f(\boldsymbol\decision) = \sum\nolimits_{i \in [n]} U_i(\boldsymbol\decision) = U(\boldsymbol\decision) \\
    \text{s.t.} \quad & \decision_i \in [0,1]^{m_i}, \boldsymbol1^\top \decision_i = 1 & \forall i \in [n] \\
    & \decision_i = \arg\max_{\decision \in \decisionset_i} f_i(\decision_i, \decision_{-i}, \boldsymbol\environment) & \forall i \in [n] \\
    & \environment(\boldsymbol\decision) \leq B
\end{align*}
where $f_i$ is defined in Equation~\ref{eqn:qre-individual-payoff}.
By converting the inner-level optimization problem to its KKT conditions, we can rewrite the optimization problem as:
\begin{align*}
    \min_{\boldsymbol\environment, \boldsymbol\decision, \boldsymbol\lambda, \boldsymbol\mu, \boldsymbol\nu} & - f(\boldsymbol\decision) = - U(\boldsymbol\decision) \\
    \text{s.t.} \quad & \decision_i, \quad \boldsymbol1^\top \decision_i = 1 & \forall i \in [n] \\
    & \lambda_i, \mu_i \in \R_{\geq 0}^{m_i}, \nu_i \in \R & \forall i \in [n] \\
    & \lambda_{i,j} x_{i,j} = 0 & \forall i \in [n], j \in [m_i] \\
    & \mu_{i,j} (1 - x_{i,j}) = 0 & \forall i \in [n], j \in [m_i] \\
    & - \nabla_{x_i} f_i - \lambda_{i} + \mu_{i} + \nu_i \boldsymbol1 = 0 & \forall i \in [n] \\
    & \environment(\boldsymbol\decision) \leq B
\end{align*}
We add dual variables $\lambda_i, \mu_i$ to the inequality constraints $\decision_{i,j} \geq 0$ and $\decision_{i,j} \leq 1$ respectively.
We also add dual variables $\nu_i$ to the equality constraints $\boldsymbol 1^\top x_i = 1$.
We can explicitly write down the gradient:
\begin{align}
    \nabla_{\decision_i} f_i(\decision_i, \decision_{-i}, \boldsymbol\environment) = (U_i + \boldsymbol\environment_i) (\decision_{-i}) - \sum\nolimits_{j} (1 + \log x_{ij})/\lambda
\end{align}
where $\lambda$ here is a specific constant (different from the Lagrangian multipliers), which is chosen to be $1$ in our implementation.

\subsection{Stackelberg Security Games With Multiple Defenders}
The followers' objectives are defined by:
\begin{align}
    f_i(\boldsymbol\decision, \environment) = \sum\nolimits_{t \in T_i} (U_{i,t} + \environment_{i,t}) (1 - y_t) p_t ,
\end{align}
where $U_{i,t}$ is the loss received by defender~$i$ when target $t$ is successfully attacked, and $\environment_{i,t}$ is the corresponding reimbursement provided by the leader to remedy the loss.
We define $y_t \coloneqq 1 - \prod\nolimits_{i} (1 - \decision_{i,t})$ to denote the effective coverage of target $t$, representing the probability that target $t$ is protected under the overlapping protection patrol plan $\boldsymbol\decision$.
Given the effective coverage of all targets, we assume the attacker attacks target $t$ with probability $p_t = e^{-\omega y_t + a_t}/(\sum\nolimits_{s \in T} e^{-\omega y_s + a_s}) $, where $a_t \in \R$ is a known attractiveness value and $\omega \geq 0$ is a scaling constant.

The leader's objective and constraint are respectively defined by:
\begin{align}
    f(\boldsymbol\decision, \environment) &= \sum\nolimits_{t \in T} U_{t} (1 - y_t) p_t \nonumber \\ g(\boldsymbol\decision, \environment) &= \left( \sum\nolimits_{i,t} \environment_{i,t} (1 - y_t) p_t \right) - B \leq 0, \nonumber
\end{align}
where $U_t < 0$ is the penalty for the leader when target $t$ is attacked without any coverage.

\paragraph{Bilevel optimization formulation}
Similarly, we can also write down the bilevel optimization formulation of the Stackelberg security games with multiple defenders as:
\begin{align*}
    \max_{\boldsymbol\environment} \quad & f(\boldsymbol\decision) = \sum\nolimits_{t \in T} U_{t} (1 - y_t) p_t \\
    \text{s.t.} \quad & \decision_{i,t} \in [0,1] & \forall i \in [n], t \in T_i \\
    & y_{t}, p_t \in \R & \forall t \in T \\
    & \sum\nolimits_{t \in T_i} \decision_{i,t} = b_i & \forall i \in [n] \\
    & y_t = 1 - \prod\nolimits_{i: t \in T_i} (1 - \decision_{i,t}) & \forall t \in T \\
    & p_t = \frac{e^{-\omega y_t + a_t}}{\sum\nolimits_{s \in T} e^{-\omega y_s + a_s}} & \forall t \in T \\
    & \decision_i = \arg\max_{\decision \in \decisionset_i} f_i(\decision_i, \decision_{-i}, \boldsymbol\environment) & \forall i \in [n] \\
    & \sum\limits_{i,t} \left( \environment^u_{i,t} (1 \! - \! y_t) p_t \! + \! \environment^c_{i,t} y_t p_t \right) \! \leq \! B
\end{align*}
where $p_t$ is the probability that attacker will attack target $t$ under protect scheme $\boldsymbol\decision$ and the resulting $\boldsymbol y$. The function $f_i$ is defined in by:
\begin{align}\label{eqn:ssg-inidividual-payoff}
    f_i(\boldsymbol\decision, \environment) = \sum\nolimits_{t \in T_i} (U_{i,t} + \environment_{i,t}) (1 - y_t) p_t.
\end{align}

This bilevel optimization problem can be reformulated into a single level optimization problem if we assume all the individual followers break ties (equilibria) in favor of the leader, which is given by:
\begin{align*}
    \max_{\boldsymbol\environment, \boldsymbol\decision, \boldsymbol\lambda, \boldsymbol\mu, \boldsymbol\nu} ~ & \sum\nolimits_{t \in T} U_{t} (1 - y_t) p_t \\
    \text{s.t.} ~ & \decision_{i,t} \in [0,1] & \forall i \in [n], t \in T_i \\
    & y_{t}, p_t \in \R & \forall t \in T \\
    & \sum\nolimits_{t \in T_i} \decision_{i,t} = b_i & \forall i \in [n] \\
    & y_t = 1 - \prod\nolimits_{i: t \in T_i} (1 - \decision_{i,t}) & \forall t \in T \\
    & p_t = \frac{e^{-\omega y_t + a_t}}{\sum\nolimits_{s \in T} e^{-\omega y_s + a_s}} & \forall t \in T \\
    & \lambda_{i,t}, \mu_{i,t} \in \R_{\geq 0}, \nu_i \in \R_{\geq 0} & \forall i \in [n], t \in T_i \\
    & \lambda_{i,t} x_{i,t} = 0 & \forall i \in [n], t \in T_i \\
    & \mu_{i,t} (1 - x_{i,t}) = 0 & \forall i \in [n], t \in T_i \\
    & - \nabla_{x_i} f_i - \lambda_{i} + \mu_{i} + \nu_i \boldsymbol1 = 0 & \forall i \in [n] \\
    & \sum\limits_{i,t} \left( \environment^u_{i,t} (1 \! - \! y_t) p_t \! + \! \environment^c_{i,t} y_t p_t \right) \! \leq \! B
\end{align*}
Similarly, we add dual variables $\lambda_{i,t}, \mu_{i,t}, \nu_{i}$ to constraints $x_{i,t} \geq 0$, $x_{i,t} \leq 1$, and $\sum\nolimits_{t \in T_i} x_{i,t} = b_i$.

\subsection{Cyber Insurance Games}
The followers' objectives are defined by:
\begin{align}\label{eqn:cyber-individual-payoff}
    f_i(\boldsymbol\decision, \environment) = - c_i \decision_i - \rho_i - (L_i - I_i) q_i - \gamma |L_i - I_i| \sqrt{q_i (1 - q_i)},
\end{align}
where $c_i$ is the unit cost of the protection $\decision_i$ and $L_i$ is the loss when the computer is attacked.
The insurance plan offered to agent~$i$ is denoted by $(\rho_i, I_i)$, where $\rho_i$ is the fixed premium paid to enroll in the insurance plan and $I_i$ is the compensation received when the computer is attacked.

We assume the computer is attacked with a probability $q_i$, where $q_i = \sigma(- \sum\nolimits_{j=1}^n w_{ij} \decision_j + v L_i)$ with $\sigma$ being sigmoid function, a matrix $W = \{w_{ij} > 0\}_{i,j \in [n]}$ to represent the interconnectedness between agents, $v \geq 0$ to reflect the attacker's preference over high-value targets, and lastly it depends on the loss $L_i$ incurred by agent $i$ when attacked.
This attack probability is a smooth non-convex function, which makes the reformulation approach hard and the non-convexity can lead to multiple equilibria reached by the followers.

The last term in Equation~\ref{eqn:cyber-individual-payoff} is the risk penalty to agent~$i$. This term is the standard deviation of the loss received by agent~$i$. We assume the agent is risk averse and thus penalized by a constant time of the standard deviation.

On the other hand, the leader's objective is defined by:
\begin{align}
    f(\boldsymbol\decision, \environment) &= \sum\nolimits_{i=1}^n - I_i q_i + \rho_i \nonumber
\end{align}
where the leader's objective is simply the total revenue received by the insurer, which includes the premium collected from all agents and the compensation paid to all agents.

The constraints are the individual rationality of each agent, where the customized insurance plan needs to incentivize the agent to purchase the insurance plan. In other words, the compensation $I_i$ and premium $\rho_i$ must incentivize agents to purchase the insurance plan by making the payoff with insurance no worse than the payoff without.
\begin{align}
    g_i(\boldsymbol\decision,\environment) &= \left(- c_i \decision_i - L_i q_i - \gamma L_i \sqrt{q_i(1 - q_i)} \right) - f_i(\boldsymbol\decision, \environment) \leq 0 \nonumber.
\end{align}

\paragraph{Bilevel optimization reformulation}
The bilevel optimization formulation for the cyber insurance domain with an external insurer is given by:
\begin{align*}
    \max_{\boldsymbol\environment} \quad & f(\boldsymbol\decision) = \sum\nolimits_{i = 1}^n -I_i q_i + \rho_i \\
    \text{s.t.} \quad & \decision_{i} \in [0,\infty) & \forall i \in [n] \\
    & q_i = \sigma \left(- \sum\nolimits_{j=1}^n w_{ij} x_j + v L_i \right) & \forall i \in [n] \\
    & \decision_i = \arg\max_{\decision'_i \in \decisionset_i} f_i(\decision'_i, \decision_{-i}, \environment) & \forall i \in [n] \\
    & - \! c_i x_i \! - \! L_i q_i \! - \! \gamma L_i \sqrt{q_i (1 \! - \! q_i)} \leq f_i(\boldsymbol\decision, \environment) & \forall i \in [n]
\end{align*}
where $f_i(\boldsymbol\decision, \environment) = -c_i x_i - \rho_i - (L_i - I_i) q_i - \gamma \norm{L_i - I_i} \sqrt{q_i (1 - q_i)}$.

Reformulating this bilevel problem into a single level optimization problem, we have:
\begin{align*}
    \max_{\boldsymbol\environment, \boldsymbol\decision, \lambda} \quad & f(\boldsymbol\decision) = \sum\nolimits_{i = 1}^n -I_i q_i + \rho_i \\
    \text{s.t.} \quad & \decision_{i} \in [0,\infty), \lambda_{i} \in [0,\infty) & \forall i \in [n] \\
    & q_i = \sigma \left(- \sum\nolimits_{j=1}^n w_{ij} x_j + v L_i \right) & \forall i \in [n] \\
    & \decision_i \lambda_i = 0 & \forall i \in [n] \\
    & - \! c_i x_i \! - \! L_i q_i \! - \! \gamma L_i \sqrt{q_i (1 \! - \! q_i)} \leq f_i(\boldsymbol\decision, \environment) & \forall i \in [n] \\
    & - \nabla_{\decision_i} f_i - \lambda_i = 0 & \forall i \in [n]
\end{align*}
with dual variables $\lambda_{i}$ for the $x_{i} \geq 0$ constraint.

\section{Experimental Setup}\label{sec:experimental-setup}
For reproducibility, we set the random seeds to be from $1$ to $30$ for NSGs and cyber insurance games, and from $1$ to $100$ for SSGs.

\subsection{Normal-Form Games}
In NFGs, we randomly generate the payoff matrix $U_i \in \R^{m_1 \times m_2 \times \cdots \times m_n}$ of follower $i$ with each entry of the payoff matrix randomly drawn from a uniform distribution $U(0, 10)$.
We assume there are $n=3$ followers.
Each follower has three pure strategies to use $m_i = m = 3$ for all $i$.
The risk aversion penalty constant is set to be $\lambda = 1$.

\subsection{Stackelberg Security Games}
In SSGs, we randomly generate the penalty $U_{i,t} < 0$ of each defender $i$ associated to each target $t \in T_i \subset T$ from a uniform distribution $U_{i,t} \sim U(-10, 0)$.
The leader's penalty $U_t < 0$ is also generated from the same uniform distribution $U_t \sim U(-10, 0)$.
We assume there are $n=5$ followers in total.
There are $|T| = 100$ targets and each follower is able to protect $|T_i| = m = 50$ targets randomly sampled from all targets.
Each follower can spend at most $b_i = 10$ effort on the available targets.
The attractiveness values $a_t$ used to denote the attacker's preference is randomly generated from a normal distribution $a_t \in \mathcal{N}(0,1)$ with $0$ mean and standard deviation $1$.
The scaling constant is set to be $\omega = 5$.

\subsection{Cyber Insurance Games}
In cyber insurance games, for each follower $i$, we generate the unit protection cost $c_i$ from a uniform distribution $c_i \sim U(5,10)$ , and the incurred loss $L_i$ from a uniform distribution $L_i \sim U(50,100)$.
We assume there are in total $n=10$ followers.
Each follower can only determine their own investment and thus $m=1$.
The entry of the correlation matrix $W \in \R^{n \times n}$ is generated from uniform distributions $W_{i,j} \sim U(0,1)$ if $i \neq j$, and $W_{i,j} \sim U(1,2)$ if $i = j$ to reflect the higher dependency on the self investments.
We choose the risk aversion constant $\gamma$ to be $\gamma = 0.01$.

\section{Computing Infrastructure}\label{sec:computing-infrastructure}
All experiments except VI experiments were run on a computing cluster, where each node is configured with 2 Intel Xeon
Cascade Lake CPUs, 184 GB of RAM, and 70 GB of local scratch space.
VI experiments require a Knitro license and were run on a machine with i9-7940X CPU @ 3.10GHz with 14 cores and 128 GB of RAM.
Within each experiment, we did not implement parallelization, so each experiment was purely run on a single CPU core.

\end{document}